\documentclass[12pt]{amsart}
\usepackage{amsmath,amsfonts,latexsym,graphicx,amssymb,url}
\usepackage{hyperref,pdfsync}
\usepackage{amsmath,txfonts,pifont,bbding,pxfonts,manfnt}
\usepackage[active]{srcltx}
\usepackage{cases,soul}
\usepackage{wasysym,pstricks, enumerate}
\usepackage{lscape,color,datetime,subfigure,graphicx,wrapfig}

\newcommand{\dist}{\text{dist}} 

\newcommand{\R}{{\mathbb R}} 

\renewcommand{\(}{\left(}
\renewcommand{\)}{\right)}
\newtheorem{theorem}{Theorem}[section]
\newtheorem{claim}{Claim}[section]

\newtheorem{proposition}[theorem]{Proposition}

\newtheorem{remark}[theorem]{Remark}

\title[Chromatic Aberration in Metalenses ~~~\today]{Chromatic Aberration in Metalenses}
\author[C. E. Guti\'errez and A. Sabra ]{Cristian E. Guti\'errez and Ahmad Sabra}
\thanks{\today}
\address{Department of Mathematics\\Temple University\\Philadelphia, PA 19122}
\email{gutierre@temple.edu}
\address{Faculty of Mathematics, Informatics, and Mechanics,
University of Warsaw, Poland \linebreak
Current Address: American University of Beirut, P. O. Box 11-0236 / Riad El-Solh / Beirut 1107 2020, Lebanon}
\email{asabra@aub.edu.lb}

\begin{document}
\begin{abstract}
This paper provides a mathematical approach to study chromatic aberration in metalenses. It is shown that radiation of a given wavelength is refracted according to a generalized Snell's law which together with the notion of envelope yields the existence of phase discontinuities. This is then used to establish a quantitative measure of dispersion in metalenses concluding that in the visible spectrum it has the same order of magnitude as for standard lenses.    
\end{abstract}
\maketitle
\tableofcontents

\section{Introduction}
The shaping of the light wavefronts with standard lenses relies
on gradual phase changes accumulated along the optical path. 
Metalenses are ultra thin surfaces which
use nano structures to focus light in imaging.
They introduce
abrupt phase shifts over the scale of the wavelength along the
optical path to bend light in unusual ways.  These nano structures are composed of arrays of
tiny pillars, rings, and other arrangements of materials, which work together to manipulate
light waves as they pass by. The subject of metalenses is a current 
important area
of research, one of the nine runners-up for Science's Breakthrough of the Year
2016 \cite{science-runner-ups-2016}, and is potentially useful in imaging applications.
Metalenses are thinner than a sheet of paper and far lighter than
glass, and they could revolutionize optical imaging devices from microscopes to virtual reality
displays and cameras, including the ones in smartphones, see  
\cite{science-runner-ups-2016} and \cite{metalenses-cameras}.

Polychromatic light is the superposition of radiation with more than one wavelength, and when it passes  
from an homogeneous medium to another, each wavelength is refracted at a different angle causing the phenomenon of dispersion. 
The failure to focus radiation composed of several wavelengths into one point is called chromatic aberration. 
It is our purpose in this paper to analyze this phenomenon in metalenses and give a measure of the dispersion.
The design of metalenses so that chromatic aberration is controlled is an active area of research \cite{aieta-capasso:generalizedrefractionfermat},\cite{WangTsai:Broadband}, and \cite{groever-capasso:substrate aberration}.
As well as for the design of standard lenses, this is an important challenge, see \cite{lalanne-chavel} and \cite{banerji---Menon}. 
These thin surfaces are made with sub wavelength phase shifters, creating a phase discontinuity on the interface yielding to a more general law of refraction, see Section \ref{sec:Snell law}. 
These components can be also be arranged to design graphene-based metasurfaces that can be actively tuned between different regimes of operation, see \cite{biswas-gutierrez-low-2018}. 
These challenging tasks have been recently an active area of research in Optical Engineering, see the comprehensive reviews \cite{yu-capasso:flat optics}, \cite{Chen-Yu:Review}, and \cite{zhu-kutznetsov-engheta}.

Mathematically, a metalens is a pair $(\Gamma,\phi)$ where $\Gamma$ is a surface in 3-d space given by the graph of a $C^2$ function $u$, and $\phi$ is a $C^1$ function defined in a small neighborhood of $\Gamma$ called phase discontinuity. 
In this paper,  we first study the existence of metalenses refracting collimated monochromatic radiation into an arbitrary field of directions. We next analyze the behavior of those metalenses when the incoming radiation is polychromatic. More precisely, we are 
given a set $\Omega\subset \mathbb R^2$, a function $u\in C^2(\Omega)$ and a $C^1$ field of unit directions 
${\bf m}(x)$. Let $\Gamma=\{(x,u(x)), x\in \Omega\}$ be the graph of $u$ and assume it is surrounded by vacuum. Our goal is to find a $C^1$ function $\phi$ defined in a neighborhood of $\Gamma$ so that for every $x\in \Omega$, each vertical monochromatic ray with wavelength $\lambda_0$ emitted from $x\in \Omega$ is refracted at $(x,u(x))$ into the direction ${\bf m}(x)$; see for example Figure \ref{fig:pic for m(x)}. 
We find in Section \ref{sec:general field of directions} necessary and sufficient conditions on the field ${\bf m}$ and the function $u$ so that the phase discontinuity $\phi$ exists in a neighborhood of $(x_0,u(x_0))$ for each $x_0\in \Omega$, see Theorem \ref{thm:metalens  existence}. 
To do this, using first the generalized Snell law from Section \ref{sec:Snell law} we see that the values $\nabla \phi(x,u(x))$
are determined. We next apply the method from Section \ref{sec:general extension method} to extend $\phi$ to a neighborhood of $(x,u(x))$. This construction is then applied when ${\bf m}$ is a constant field, and when ${\bf m}$ points toward a point $P_0=(0,0,p_0)$, see Sections \ref{subsec:constant}, and \ref{subsec:point target}.
If $\Gamma$ is the graph of any smooth radial function, and ${\bf m}$ points towards $P_0=(0,0,p_0)$, then 
there exists also a phase $\phi$ in a neighborhood of $(0,u(0))$ such that the metalens $(\Gamma,\phi)$ focuses 
all vertical rays with wavelength $\lambda_0$ into $P_0$; see Remark \ref{rmk:radial Gamma existence of phase}.
We then study the behavior of these metalenses when the 
incoming rays have wavelength $\lambda$ instead of $\lambda_0$. We prove that that if $\Gamma$ is a horizontal plane, Section \ref{sec:chromatic 
aberration plane}, or more generally if $\Gamma$ is the graph of a radial function, Section \ref{subsec:aberration for general radial u}, then all vertical rays with wavelength  $\lambda$ 
are focused into the points $P_{\lambda}(x)=(0,0,p_{\lambda}(x))$ depending on $|x|$. We then estimate in both cases the chromatic 
aberration $\left|P_0-P_{\lambda}(x)\right|$; see Theorem \ref{thm:planar case error} for the planar case and  \eqref{eq:order of magnitude for radial case} when $u$ is radial. 

Finally, in Section \ref{sec:standard lens}, we analyze a similar problem with a standard 
lens. 
In this case, it is known that the surfaces separating two media $n_1$ and $n_2$ with $n_1>n_2$ that focus 
vertical radiation in medium $n_1$ with wavelength $\lambda$ into a fixed point in medium $n_2$ are hyperboloids, see \cite{gutierrez-huang:reshaping light beam}. 
It is then shown that the error for chromatic aberration in metalenses for visible wavelengths is of the same order of magnitude as for standard lenses.

Finally, we mention that the phase discontinuity functions needed to design metalenses for various refraction and reflection problems with prescribed distributions of energy satisfy partial differential equations of Monge-Amp\`ere type which are derived and 
studied in \cite{gutierrez-pallucchini:metasurfacesandMAequations}. 
These type of equations appeared also naturally in solving problems involving aspherical lenses, see \cite{gutierrez-huang:reshaping light beam}, \cite{abedin-gutierrez-tralli:regularity-refractors},  \cite{gutierrez-sabra:asphericallensdesignandimaging}, \cite{gutierrez-sabra:freeformgeneralfields}, and references therein.

\section{Generalized Snell's law}\label{sec:Snell law}
The refractive index of a medium corresponding to an electromagnetic wave with frequency $\omega$ is given by $n=\dfrac{c}{v}$, where $c$ is the velocity of the wave in vacuum and $v$ its apparent velocity in the medium. 
$v$ depends on the wave and on the material \cite[Vol. 1, Chap. 31]{feynman-lectures-on-physics}, see formula (31.19) there for $n$ that highlights the dependence of $n$ on $\omega$ and on the material. 
See also \cite[Chap. 3]{Hecht-optics-book} and \cite[Sec. 2.3.4]{book:born-wolf}.
Note that the frequency $\omega$ is a characteristic of the wave and is independent of the material traversed. The wavelength of the wave in the medium is the distance between two consecutive crests and is given by $\lambda=\dfrac{2\pi}{\omega} v$. 
The wave number $k=\dfrac{2\pi}{\lambda}$ is the number of waves per unit distance,
 $\lambda_0$ denotes the wavelength of a wave with the same frequency $\omega$ when traveling in vacuum, and $k_0=\dfrac{2\pi}{\lambda_0}$ is the corresponding wave number. We have $\lambda_0=\dfrac{2\pi c}{\omega}$, then 
\begin{equation}\label{eq:formulas for refractive indices and wavelenghts}
n=\dfrac{c}{v}=\dfrac{2\pi c}{\lambda\omega} =\dfrac{\lambda_0}{\lambda}=\dfrac{k}{k_0}.
\end{equation}

In \cite{gutierrez-pallucchini-stachura:nonflatmetasurfaces}, Snell's law has been generalized for metasurfaces using wave fronts introducing a phase shift along the optical path. 
Using Fermat principle, we next derive here this law in a form needed for our purposes. Assume $\Gamma$ is a surface,   separating media $I$ and $II$, and given by the level set $\psi(x,y,z)=0$, with $\psi$ a $C^1$ function. Assume a phase discontinuity $\phi$ is defined in a neighborhood of the surface $\Gamma$. A light wave with frequency $\omega$ emitted from a point $A$ in medium $I$ strikes $\Gamma$ at the point $P(x,y,z)$ and is then refracted into the point $B$ in medium $II$. 
Let $n_1$ and $n_2$ be the refractive indices of media $I$ and $II$ corresponding to $\omega$, and let $k_1, k_2$ be the corresponding wave numbers.  By Fermat's principle of stationary phase we have that $P(x,y,z)$ is a critical point to the function 
$$k_1\, |P-A|+k_2\,|B-P| -\phi(P),$$
and since $P\in \Gamma$, then we also have that $\psi(P)=0$.
Therefore by Lagrange multipliers, we have that 
$$\nabla\left(k_1\, |P-A|+k_2\,|B-P| -\phi(P)\right)\times \nabla \psi(P)=0.$$
Since $\nabla \psi$ is parallel to the normal $\nu$ to $\Gamma$ at the point $P$, and
$$\nabla\left(k_1\, |P-A|+k_2\,|B-P| -\phi(P)\right)=k_1\dfrac{{P-A}}{|P-A|}-k_2\dfrac{{B-P}}{|B-P|}-\nabla \phi(P),$$
denoting the unit directions of the incident and refracted rays by ${\bf x}=\dfrac{{P-A}}{|P-A|}$ and ${\bf m}=\dfrac{{B-P}}{|B-P|}$, respectively, we then get the following form of Snell's law for the metalens $(\Gamma,\phi)$
\begin{equation}\label{eq:snell law for wave vectors}
\left(k_1{\bf x}-k_2{\bf m}\right)\times \nu=\nabla \phi\times \nu.
\end{equation}
From \eqref{eq:formulas for refractive indices and wavelenghts}, $k_i=n_ik_0$ and so \eqref{eq:snell law for wave vectors} can be re-written as
\begin{equation*}
 k_0(n_1 {\bf x} - n_2 {\bf m})\times \nu=\nabla \phi\times \nu,
 \end{equation*}
 We then obtain
 \begin{equation}\label{eq:SnellLaw}
 \(n_1 {\bf x} - n_2 {\bf m}\)\times \nu=\(\nabla \(\lambda_0\,\phi/2\pi\)\)\times \nu.
 \end{equation} 
 This is clearly equivalent to 
 \begin{equation}\label{eq:SnellLaw bis}
 n_1 {\bf x} - n_2 {\bf m}=\mu\,\nu + \nabla \(\lambda_0\,\phi/2\pi\), 
 \end{equation}
for some $\mu$.
From \cite[Formula (11)]{gutierrez-pallucchini-stachura:nonflatmetasurfaces}, the value of $\mu$ can be calculated:
\begin{equation}\label{eq:formula for mu}
\mu=\left(n_1{\bf x}-\nabla \(\lambda_0\,\phi/2\pi\)\right)\cdot \nu-\sqrt{n_2^2-\left(\left|n_1{\bf x}-\nabla \(\lambda_0\,\phi/2\pi\)\right|^2-\left[\left(n_1{\bf x}-\nabla \(\lambda_0\,\phi/2\pi\)\right)\cdot \nu\right]^2\right)}\,.
\end{equation}
In order to avoid total internal reflection, the term inside the square root in \eqref{eq:formula for mu} must be non-negative, that is,
\begin{equation}\label{eq:internal reflection}
\left[\left({\bf x}-\dfrac{\nabla\(\lambda_0\,\phi/2\pi\)}{n_1}\right)\cdot \nu\right]^2\geq \left|{\bf x}-\dfrac{\nabla\(\lambda_0\,\phi/2\pi\)}{n_1}\right|^2-\left(\dfrac{n_2}{n_1}\right)^2.
\end{equation}

\section{Construction of gradient fields using the notion of envelope}\label{sec:general extension method}
\setcounter{equation}{0}
Let $\Omega\subset \R^2$ be an open simply connected domain, $u\in C^2(\Omega)$, and 
let $\Gamma$ be the graph of $u$, i.e., $\Gamma=\{(x,u(x)):x\in \Omega\}$; $x=(x_1,x_2)$.
A field $V=(V_1,V_2,V_3):\Gamma\to \R^3$ is given on $\Gamma$ such that $V(x,u(x))\in C^1(\Omega)$. 
Here we answer the following question: When is it possible to have a $C^1$ function $\phi$ defined in a neighborhood $\Gamma$ such that $\nabla \phi=V$ on $\Gamma$?

Suppose such a $\phi$ exists, that is, $\nabla \phi(x,u(x))=V(x,u(x))$, $x=(x_1,x_2)\in \Omega$, and consider $f(x)=\phi(x,u(x))$. Since $\phi=\phi(x_1,x_2,x_3)$ is differentiable and defined in a neighborhood of $\Gamma$, $f_{x_i}(x)=\phi_{x_i}(x,u(x))+ \phi_{x_3}(x,u(x))\,u_{x_i}(x)
=
V_i(x,u(x))+V_3(x,u(x))\,u_{x_i}(x)$, $i=1,2$, so
$f$ satisfies the system
\begin{equation}\label{eq:system of differential equations phi along Gamma}
\begin{cases}
f_{x_1}(x)
=
V_1(x,u(x))+V_3(x,u(x))\,u_{x_1}(x)\\
f_{x_2}(x)
=
V_2(x,u(x))+V_3(x,u(x))\,u_{x_2}(x),
\end{cases}
\end{equation}
for $x\in \Omega$.
Since $V(x,u(x))$ is $C^1$ and $u\in C^2$, it follows that $f\in C^2$ and so $f_{x_1x_2}(x)=f_{x_2x_1}(x)$ which is equivalent to
\begin{equation}\label{eq:necessary condition of field psi}
\partial_{x_2}\(V_1(x,u(x))+V_3(x,u(x))\,u_{x_1}(x)\)
=
\partial_{x_1}\(V_2(x,u(x))+V_3(x,u(x))\,u_{x_2}(x)\),\quad x\in \Omega.
\end{equation}
Therefore, if such a $\phi$ exists, then $V$ must satisfy \eqref{eq:necessary condition of field psi}.

Vice versa, if we are now given a field $V$ defined on $\Gamma$ and satisfying \eqref{eq:necessary condition of field psi},
we want to construct $\phi$ in a neighborhood of $\Gamma$ so that $\nabla \phi=V$ along $\Gamma$. Actually, we will give a sufficient condition on the field $V$ so that $\phi$ exists in a neighborhood of a point $\(x_0,u(x_0)\)\in \Gamma$.
In order to do this, we will apply the implicit function theorem, for which we assume  $V(x,u(x))$ to be $C^2$ in a neighborhood of $x_0$.
Indeed, if $V$ satisfies \eqref{eq:necessary condition of field psi} in a ball $B=B_\delta(x_0)\subset \Omega$, then the system \eqref{eq:system of differential equations phi along Gamma} is solvable in $B$, and let $f$ be a solution\footnote{There are many solutions to \eqref{eq:system of differential equations phi along Gamma}. However, given a point $x_0\in \Omega$ and a number $y_0\in \R$, there is a unique solution $f$ satisfying \eqref{eq:system of differential equations phi along Gamma} with $f(x_0)=y_0$.}.
We then have
\begin{align}
V(x,u(x))\cdot (1,0,u_{x_1}(x))&=\partial_{x_1}f(x)\label{eq:compatibility cond x}\\
V(x,u(x))\cdot (0,1,u_{x_2}(x))&=\partial_{x_2}f(x).\label{eq:compatibility cond y}
\end{align}
We seek a function $\phi$ defined in a neighborhood of the point $(x_0,u(x_0))$ so that 
$\phi(x,u(x))=f(x)$, and $\nabla  \phi(x,u(x))=V(x,u(x))$ in a neighborhood of $x_0$.

Let us then define $\mathcal S$ to be the surface in $\R^4$ given by the vector 
$$P(x)=(x,u(x),f(x)),\qquad\qquad x\in B.$$
At each point $P(x)$ consider the vector
$$N(x)=\left(V(x,u(x)),-1\right),$$
and the plane $\Pi_{x}$ passing through $P(x)$ with normal $N(x)$, that is,
$$\Pi_x=\{y=(y_1,y_2,y_3,y_4)\in \R^4: \left(y-P(x)\right)\cdot N(x)=0\}.$$
For $x\in B$ and $y\in \R^4$ define the function 
\[
F(y,x)=\left(y-P(x)\right)\cdot N(x)
\] 
and the map 

\begin{equation}\label{def:map G}
G(y,x)=\left(F(y,x), \dfrac{\partial F}{\partial x_1}(y,x) ,\dfrac{\partial F}{\partial x_2}(y,x) \right).
\end{equation}
Consider then the system of equations
\begin{equation}\label{eq:System of F near field}
G(y,x)=(0,0,0).
\end{equation}
We will construct $\phi$ by solving \eqref{eq:System of F near field} using the implicit function theorem.

\begin{claim}\label{clm:solution to the system near field} 
For each $x\in B$, the vector $(P(x),x)$ solves the system \eqref{eq:System of F near field}.
\end{claim}
\begin{proof}
Clearly, from the definition of $F$, $F(P(x),x)=0$.
From \eqref{eq:compatibility cond x}
\begin{align*}
P_{x_1}(x)\cdot N(x)
&= 
\(1,0,u_{x_1}(x),\partial_{x_1}f(x)\)
\cdot
\left(V(x,u(x)),-1\right)=0.
\end{align*}
Hence for each $(y,x)\in \R^4\times  B$ 
\begin{align*}
\dfrac{\partial F}{\partial x_1}&=-P_{x_1}(x)\cdot N(x)+(y-P(x))\cdot N_{x_1}(x)
=(y-P(x))\cdot N_{x_1}(x),
\end{align*}
and so  
$$\dfrac{\partial F}{\partial x_1}(P(x),x)=0.$$
Similarly, from \eqref{eq:compatibility cond y} $P_{x_2}(x)\cdot N(x)=0$ so we get $\dfrac{\partial F}{\partial x_2}(P(x),x)=0$, and the claim follows.
\end{proof}

Consider the point $\(P(x_0),x_0\)=\(x_0,u(x_0),f(x_0),x_0\).$ 
From Claim \ref{clm:solution to the system near field}, $G\(P(x_0),x_0\)=(0,0,0)$. 
From the differentiability assumptions on $u$ and $V(x,u(x))$,
$G$ has continuous first order partial derivatives in a neighborhood of $(P(x_0),x_0)$.  If the Jacobian determinant
\begin{equation}\label{eq:Jacobian different from zero}
\dfrac{\partial G}{\partial (y_4,x_1,x_2)}\(P(x_0),x_0\)\neq 0,
\end{equation}
then the implicit function theorem implies that 
there exist an open neighborhood $O\subseteq B\times \R$ of $\(x_0,u(x_0)\)$, an open neighborhood $W\subseteq \R\times  B$  of $(f(x_0),x_0)$, and unique $C^1$ functions $g_1,g_2,g_3:O\to W$ 
solving
\begin{equation}\label{eq:equation for the implicit function theorem}
G(y_1,y_2,y_3,g_1(y_1,y_2,y_3),g_2(y_1,y_2,y_3),g_3(y_1,y_2,y_3))=(0,0,0),
\end{equation}
for all $(y_1,y_2,y_3)\in O$.

We shall prove that the desired $\phi$ is $\phi=g_1$.
Indeed, from the continuity of $u(x)$ and $f(x)$ in $B$ there exists a neighborhood $U\subseteq B$ of $x_0$  such that for every $x\in U$, we have $(x,u(x))\in O$, and $(f(x),x)\in W$. Hence, for $x\in U$ we have $(P(x),x)\in O\times  W$, and from Claim \ref{clm:solution to the system near field} $G(P(x),x)=(0,0,0)$. 
Therefore, from the uniqueness of the functions $g_1,g_2,g_3$, it follows that $g_1(x,u(x))=f(x)$, $g_2(x,u(x))=x_1, g_3(x,u(x))=x_2$, for every $x=(x_1,x_2)\in U$. 
Then, 
it remains to show  that $\nabla g_1(x,u(x))=V(x,u(x))$ for $x\in U$, i.e., along $\Gamma$.
In fact, since for $(y_1,y_2,y_3)\in O$  
\begin{equation}\label{eq:identity for F}
F(y_1,y_2,y_3,g_1(y_1,y_2,y_3),g_2(y_1,y_2,y_3),g_3(y_1,y_2,y_3))=0,
\end{equation} 
differentiating this identity with respect to $y_1$ yields
\begin{equation}\label{eq:derivative wrt to x bis}
0=F_{y_1}+F_{y_4}(g_1)_{y_1}+F_{x_1}(g_2)_{y_1}+F_{x_2}(g_3)_{y_1}.
\end{equation}
From \eqref{eq:equation for the implicit function theorem}, $F_{x_1}=F_{x_2}=0$ at $\(y_1,y_2,y_3,g_1(y_1,y_2,y_3),g_2(y_1,y_2,y_3),g_3(y_1,y_2,y_3)\)$. 
Moreover for each $(y,x)\in \R^4\times B$, we have
\begin{align*}
F_{y_1}(y,x)&=(1,0,0,0)\cdot N(x)=V_1(x,u(x))\\
F_{y_4}(y,x)&=(0,0,0,1)\cdot N(x)=-1.
\end{align*}
Since $g_2(x,u(x))=x_1$ and $g_3(x,u(x))=x_2$ for $x=(x_1,x_2)\in U$, substituting the obtained values of $F_{x_1},F_{x_2},F_{y_1},F_{y_4}$ in \eqref{eq:derivative wrt to x bis} yields
$$V_1(x,u(x))=(g_1)_{y_1}(x,u(x)),\quad x\in U.$$
Differentiating \eqref{eq:identity for F} with respect to $y_2$ and with respect to $y_3$, and proceeding similarly, yields 
$$V_2(x,u(x))=(g_1)_{y_2}(x,u(x)),\quad V_3(x,u(x))=(g_1)_{y_3}(x,u(x)).$$
Hence $\nabla g_1(x,u(x))=V(x,u(x))$ for $x\in U$.

We have then proved the following theorem.

\begin{theorem}\label{thm:existence of phase discontinuitiy functions}
Let $x_0\in \Omega$, $u\in C^2$ in a neighborhood of $x_0$, and let $V(x,u(x))$ be a field that is $C^2$ in a neighborhood of $x_0$ satisfying \eqref{eq:necessary condition of field psi} in that neighborhood.
If \eqref{eq:Jacobian different from zero} holds, where $G$ is defined by \eqref{def:map G}, then there exists a function $\phi$ defined in a neighborhood of $(x_0,u(x_0))$ such that $V=\nabla  \phi$ in that neighborhood.

\end{theorem}

For our application, let us re write condition \eqref{eq:Jacobian different from zero} in simpler terms.
We have
$$
\dfrac{\partial G}{\partial (y_4,x_1,x_2)}\(P(x_0),x_0\)=\left|\begin{matrix}
    \dfrac{\partial F}{\partial y_4}      & \dfrac{\partial^2 F}{\partial x_1 \partial y_4}  & \dfrac{\partial^2 F}{\partial x_2\partial y_4}  \\
    \dfrac{\partial F}{\partial x_1}& \dfrac{\partial^2 F}{\partial^2 x_1}  & \dfrac{\partial^2 F}{\partial x_2\partial x_1} \\
    \dfrac{\partial F}{\partial x_2}      & \dfrac{\partial^2 F}{\partial x_1 \partial x_2}&  \dfrac{\partial^2 F}{\partial^2 x_2}
\end{matrix}
\right|_{(P(x_0),x_0)}.
$$
Since $F_{y_4}(y,x)=-1$,  
$$
\dfrac{\partial G}{\partial (y_4,x_1,x_2)}\(P(x_0),x_0\)=\left|\begin{matrix}
    -1      & 0  & 0  \\
    F_{x_1}& F_{x_1x_1}  & F_{x_1x_2} \\
    F_{x_2}      & F_{x_2x_1}&  F_{x_2x_2}
\end{matrix}
\right|_{\(P(x_0),x_0\)}
=-\left|\begin{matrix}
   F_{x_1x_1}  & F_{x_1x_2} \\
   F_{x_2x_1}&  F_{x_2x_2}
\end{matrix}
\right|_{\(P(x_0),x_0\)}.
$$
From the proof of Claim \ref{clm:solution to the system near field}, $F_{x_i}(y,x)=(y-P(x))\cdot N_{x_i}(x)$, $i=1,2$, 
and differentiating yields $F_{x_ix_j}(y,x)=-P_{x_j}\cdot N_{x_i}+(y-P)\cdot N_{x_ix_j}$.
Therefore \eqref{eq:Jacobian different from zero} becomes
\begin{equation}\label{eq:determinant simpler writing}
\dfrac{\partial G}{\partial (y_4,x_1,x_2)}\(P(x_0),x_0\)=-\left|\begin{matrix}
   P_{x_1}(x_0)\cdot N_{x_1}(x_0)  & P_{x_1}(x_0)\cdot N_{x_2}(x_0) \\
   P_{x_2}(x_0)\cdot N_{x_1}(x_0)&  P_{x_2}(x_0)\cdot N_{x_2}(x_0)
\end{matrix}\right|\neq 0.
\end{equation}

\section{Existence of phase discontinuities focusing a collimated beam into a given set of directions}\label{sec:general field of directions}
\setcounter{equation}{0}
We are given a $C^1(\Omega)$ unit field with non negative vertical component,
$${\bf m}(x)= (m_1(x),m_2(x),m_3(x)),\quad x\in \Omega, \quad |{\bf m}(x)|=1\quad ,m_3(x)\geq 0,
$$
corresponding to the set of directions where we want the radiation to be steered.
Assume the medium surrounding $\Gamma$ is vacuum; $\Gamma$ is the graph of a function $u\in C^2(\Omega)$. We consider a vertical parallel beam of monochromatic light 
of wavelength $\lambda_0$. From each $x\in \Omega$ a ray with direction ${\bf e}
=(0,0,1)$ strikes $\Gamma$ at the point $(x,u(x))$. 
The goal of this section is to apply the results from Section \ref{sec:general extension method} to show existence of a phase  
discontinuity $\phi\in C^1$ defined in a neighborhood of $\Gamma$ so that $\nabla \phi$ is tangential to $\Gamma$, and for each $x\in \Omega$, the vertical ray 
with direction ${\bf e}$ is refracted by the metalens $(\Gamma,\phi)$ at $(x,u(x))$ into the direction ${\bf m}(x)$; see Figure \ref{fig:pic for m(x)}.

Let $\nu:=\nu(x)=\dfrac{(-\nabla u,1)}{\sqrt{1+|\nabla u|^2}}$ be the outer unit normal to $\Gamma$ at each point $\(x,u(x)\)$. Since $m_3\geq 0$, 
it follows that $m(x)\cdot \nu\geq 0.$ 
A tangential phase discontinuity $\phi$ means $\nabla \phi(x,u(x))\cdot \nu=0$ for each $x\in \Omega$. 

We shall prove the following theorem.
\begin{theorem}\label{thm:metalens  existence}
If a tangential phase discontinuity $\phi$ solving the problem above exists, then 
\begin{equation}\label{eq:nabla phi general surface}
\nabla \phi(x,u(x))=\dfrac{2\pi}{\lambda_0}\({\bf e}-{\bf m}-\(\({\bf e}-{\bf m}\)\cdot \nu\)\,\nu \):=V(x,u(x)),
\end{equation}
and ${\bf m}$ satisfies the condition
\begin{equation}\label{eq:curl condition on m}
\left({\bf m}(x)\cdot (1,0,u_{x_1})\right)_{x_2}=\left({\bf m}(x)\cdot (0,1,u_{x_2})\right)_{x_1},
\end{equation}
for $x\in \Omega$.

Conversely, suppose ${\bf m}$ is $C^2$ and satisfies \eqref{eq:curl condition on m}, 
let $h\in C^2(\Omega)$ satisfying 
\begin{equation}\label{eq:gradient of h in terms of m}
\nabla h(x)=\left({\bf m}(x)\cdot (1,0,u_{x_1}),{\bf m}(x)\cdot (0,1,u_{x_2})\right),
\end{equation}
and let $V(x,u(x))$ be given by \eqref{eq:nabla phi general surface} with $u\in C^3(\Omega)$.
If
\begin{equation}\label{eq:determinant condition}
\det\left( D^2h-\left\{{\bf e}\cdot {\bf m}+\dfrac{({\bf e}-{\bf m})\cdot \nu}{\sqrt{1+u_x^2+u_y^2}}\right\}D^2u\right)\neq 0,
\end{equation}
at some $x_0\in \Omega$, then there exist a neighborhood $W$ of $\(x_0,u(x_0)\)$ and a function $\phi$ defined in $W$  
satisfying 
\begin{equation}\label{eq:Formula for phi general surface}
\phi(x,u(x))=\dfrac{2\pi}{\lambda_0}\left(u(x)-h(x)\right)+C,
\end{equation}
for some constant $C$, with
$\nabla \phi(x,u(x))=V(x,u(x))$ for $x$ in a neighborhood of $x_0$. 
In addition, 
for each $\(x,u(x)\)\in W$ the metalens $(\Gamma,\phi)$ refracts the ray with direction ${\bf e}$ striking $\Gamma$ at $(x,u(x))$ into the direction ${\bf m}(x)$.
\end{theorem}

\begin{proof}
To prove the first part of the theorem,
from the Snell law \eqref{eq:SnellLaw} with $n_1=n_2=1$, the phase must satisfy 
\[
\({\bf e}-{\bf m}-\(\lambda_0/2\pi\)\nabla \phi\)\times \nu=0,
\]
at all points $(x,u(x))$.
Let us calculate $\nabla \phi$ from this expression. 
Taking cross product again with $\nu$ yields
\[
\nu\times \(\({\bf e}-{\bf m}-\(\lambda_0/2\pi\)\nabla \phi\)\times \nu\)=0.
\]
Since $\phi$ is tangential to $\Gamma$, $\nabla \phi\cdot \nu=0$, $\nu\cdot \nu=1$,  and for all vectors $a,b,c$ we have $a\times (b\times c)=b\,(a\cdot c)-c\,(a\cdot b)$, we obtain \eqref{eq:nabla phi general surface}.
Hence $V\cdot \(1,0,u_{x_1}\)=\dfrac{2\pi}{\lambda_0}\(-{\bf m}\cdot \(1,0,u_{x_1}\)+u_{x_1}\)$ and 
$V\cdot \(0,1,u_{x_2}\)=\dfrac{2\pi}{\lambda_0}\(-{\bf m}\cdot \(0,1,u_{x_2}\)+u_{x_2}\)$.
Therefore, from the existence of $\phi$,  
\eqref{eq:necessary condition of field psi} holds and consequently \eqref{eq:curl condition on m} also.

To prove the converse, we apply Theorem \ref{thm:existence of phase discontinuitiy functions} to $V$ given in \eqref{eq:nabla phi general surface}. Since $\bf m$ is $C^2$ and $u$ is $C^3$, $V$ is $C^2$.
That $V$ satisfies condition \eqref{eq:necessary condition of field psi} follows from \eqref{eq:curl condition on m}.
Let us verify \eqref{eq:Jacobian different from zero}.
Let $f=\dfrac{2\pi}{\lambda_0}\(u-h\)$, so $f$ solves \eqref{eq:system of differential equations phi along Gamma}. 
From the definition of $G$ in \eqref{def:map G}, the vectors $P$ and $N$ are
\begin{align*}
P(x)=\(x,u(x),f(x)\);\qquad N(x)=\(V(x,u(x)),-1\),
\end{align*}
so to show \eqref{eq:Jacobian different from zero} it is enough to show \eqref{eq:determinant simpler writing}.
We have  
\begin{align*}
P_{x_1}\cdot N_{x_1}
&=-\dfrac{2\pi}{\lambda_0}(1,0,u_{x_1},f_{x_1})\cdot \left({\bf m}_{x_1}+[(({\bf e}-{\bf m})\cdot \nu)\nu]_{x_1},0\right)\\
&=-\dfrac{2\pi}{\lambda_0}{\bf m}_{x_1}\cdot (1,0,u_{x_1})\
-
\dfrac{2\pi}{\lambda_0}\(({\bf e}-{\bf m})\cdot \nu\)(1,0,u_{x_1})\cdot \nu_{x_1},\quad \text{since $\nu\cdot (1,0,u_{x_1})=0$}\\
&=-\dfrac{2\pi}{\lambda_0}{\bf m}_{x_1}\cdot (1,0,u_{x_1})+\dfrac{2\pi}{\lambda_0} \,\dfrac{u_{x_1x_1}}{\sqrt{1+|\nabla u|^2}} ({\bf e}-{\bf m})\cdot \nu.
\end{align*}
Notice that
$$
{\bf m}_{x_1}\cdot (1,0,u_{x_1})=({\bf m}\cdot (1,0,u_{x_1}))_{x_1}-{\bf m}\cdot (0,0,u_{x_1x_1})=h_{x_1x_1}-({\bf e}\cdot {\bf m})\,u_{x_1x_1}.
$$
Thus
$$
P_{x_1}\cdot N_{x_1}=-\dfrac{2\pi}{\lambda_0}h_{x_1x_1}+\dfrac{2\pi}{\lambda_0}\left(\dfrac{({\bf e}-{\bf m})\cdot \nu}{\sqrt{1+|\nabla u|^2}}+{\bf e}\cdot{\bf m}\right)u_{x_1x_1},
$$
and in general
\[
P_{x_i}\cdot N_{x_j}
=-\dfrac{2\pi}{\lambda_0}h_{x_ix_j}+\dfrac{2\pi}{\lambda_0}\left(\dfrac{({\bf e}-{\bf m})\cdot \nu}{\sqrt{1+|\nabla u|^2}}+{\bf e}\cdot{\bf m}\right)u_{x_ix_j}.
\]
Therefore the determinant in \eqref{eq:determinant simpler writing} is $-\left(\dfrac{2\pi}{\lambda_0}\right)^2$ times the determinant in \eqref{eq:determinant condition}, 
and consequently if \eqref{eq:determinant condition} holds we obtain \eqref{eq:Jacobian different from zero}.
Consequently, the existence of $\phi$ follows from Theorem \ref{thm:existence of phase discontinuitiy functions}.

With the phase $\phi$ obtained satisfying  \eqref{eq:Formula for phi general surface}, it remains to show that the metalens $(\Gamma,\phi)$ refracts incident rays with 
direction ${\bf e}$ into the direction ${\bf m}(x)$ at the point $(x,u(x))$ in a neighborhood of $(x_0,u(x_0))$. 
Since $\nabla \phi$ satisfies \eqref{eq:nabla phi general 
surface},  then ${\bf e}-{\bf m}-\nabla\(\lambda_0\phi/2\pi\)$ is parallel to $\nu$ and so Snell's law \eqref{eq:SnellLaw} is 
verified and consequently ${\bf e}$ is refracted into ${\bf m}$. However, to avoid total internal reflection, the ray ${\bf e}$ is refracted into the medium (vacuum) above the metasurface when ${\bf e }$ and $\nabla \phi$ satisfy condition \eqref{eq:internal reflection} with $n_1=n_2=1$, and ${\bf x}={\bf e}$, that is when 
\begin{equation}\label{eq:internal refraction bis}
\left[\({\bf e}-\nabla\(\lambda_0\phi/2\pi\)\)\cdot \nu\right]^2\geq \left|{\bf e}-\nabla\(\lambda_0\phi/2\pi\)\right|^2-1.
\end{equation}
To verify this, from \eqref{eq:nabla phi general surface} we have 
\begin{align*}
\left|{\bf e}-\nabla\(\lambda_0\phi/2\pi\)\right|^2-1
&=\left|{\bf m}+(({\bf e}-{\bf m})\cdot \nu)\nu\right|^2-1\\
&=(({\bf e}-{\bf m})\cdot\nu)^2+2(({\bf e}-{\bf m})\cdot\nu){\bf m}\cdot\nu\\
&=({\bf e}\cdot\nu)^2-({\bf m}\cdot\nu)^2\leq ({\bf e}\cdot\nu)^2,
\end{align*}
and so \eqref{eq:internal refraction bis} follows since $\nabla\phi\cdot \nu=0$.

\end{proof}

\subsection{Existence of phases when ${\bf m}$ is constant}\label{subsec:constant}
We apply the converse in Theorem \ref{thm:metalens  existence} with
${\bf m}(x):={\bf m}$ a constant vector for each $x\in \Omega$ and assume $u\in C^3$. It is clear that \eqref{eq:curl condition on m} holds.
To calculate the determinant in \eqref{eq:determinant condition}, since ${\bf m}$ is constant,
from \eqref{eq:gradient of h in terms of m} it follows that
$
h(x)={\bf m}\cdot (x,u(x))+C,
$
and $D^2h(x)=\({\bf e}\cdot {\bf m}\)D^2u(x)$. So the determinant in \eqref{eq:determinant condition} equals
\begin{equation}\label{eq:determinant when m is constant}
\det\left(
\left\{\dfrac{({\bf e}-{\bf m})\cdot \nu}{\sqrt{1+|\nabla u|^2}}\right\}D^2u\right)
=
\(\dfrac{({\bf e}-{\bf m})\cdot \nu}{\sqrt{1+|\nabla u|^2}}\)^2\,\det D^2u(x).
\end{equation}
If the last expression is not zero at $x=x_0$, 
 then \eqref{eq:determinant condition} holds and therefore a phase discontinuity $\phi$ exists in a neighborhood of $\(x_0,u(x_0)\)$.

\begin{remark}\label{rmk:planar case}\rm 
We remark that \eqref{eq:determinant condition} is only a sufficient condition for the existence of $\phi$ and it is not necessary. In fact, when $\Gamma$ is a plane, i.e., $u(x)={\bf a}\cdot x+b$ with ${\bf a}\in \R^2$, $b\in \R$, and ${\bf m}$ is a constant field, the determinant in \eqref{eq:determinant condition} equals the determinant in \eqref{eq:determinant when m is constant} which is zero. However, in this particular case we can still find a phase discontinuity $\phi$ in a neighborhood of $\Gamma$. In fact, the normal $\nu=\dfrac{\(-{\bf a},1\)}{\sqrt{1+|{\bf a}|^2}}$, and from \eqref{eq:nabla phi general surface} if $\phi$ exists it must verify
$$\nabla \phi(x,u(x))=\dfrac{2\pi}{\lambda_0}\left\{{\bf e}-{\bf m} -\left[({\bf e}-{\bf m})\cdot \dfrac{(-{\bf a},1)}{1+|{\bf a}|^2}\right](-{\bf a},1)\right\}:= {\bf v}.$$
Letting $\phi(x_1,x_2,x_3)={\bf v}\cdot (x_1,x_2,x_3)$ yields the desired phase and so the metalens $(\Gamma,\phi)$ refracts all vertical rays into the constant direction ${\bf m}$.
\end{remark}

\subsection{Existence of phases focusing all rays into one point} \label{subsec:point target}
Here we show the existence of metalenses $(\Gamma,\phi)$ that refract all vertical rays into a fixed point $P=(0,0,p)$, $p>0$. 
The set of refracted directions is 
\begin{equation}\label{eq:set of directions focusing into one point}
{\bf m}(x)=\dfrac{(-x,p-u(x))}{\sqrt{|x|^2+(p-u(x))^2}},
\end{equation}
where $u(x)<p$ for all $x$.
In this case, to show existence of phase discontinuities will use the following proposition.
\begin{proposition}\label{cor:near field}
For each function $u\in C^2$, ${\bf m}$ defined by \eqref{eq:set of directions focusing into one point}
satisfies \eqref{eq:curl condition on m}. 
Moreover, the determinant in \eqref{eq:determinant condition} can be written as
{\small
\begin{equation}\label{eq:near field condition}
 \dfrac{1}{h^2} \det\left(Id+\nabla u\otimes \nabla u-\dfrac{1}{h^2} \(x+(u-p)\,\nabla u\)\otimes \(x+(u-p)\,\nabla u\)+\dfrac{\(x,u-p-h\)\cdot (-\nabla u,1)}{1+|\nabla u|^2}D^2u\right),
\end{equation}
}
where $\nabla h=\({\bf m}\cdot \(1,0,u_{x_1}\),{\bf m}\cdot \(0,1,u_{x_2}\)\)$.
Therefore, if \eqref{eq:near field condition} is not zero at $x_0$ and $u\in C^3$, then from Theorem \ref{thm:metalens  existence}  
a phase discontinuity $\phi$ exists in a neighborhood of $(x_0,u(x_0))$.

\end{proposition}

\begin{proof}
We have
$${\bf m}(x)\cdot (1,0,u_{x_1})= \dfrac{-x_1+(p-u(x))\,u_{x_1}}{\sqrt{|x|^2+(p-u(x))^2}}=-\left\{\sqrt{|x|^2+(p-u(x))^2}\right\}_{x_1},$$
and similarly
${\bf m}(x)\cdot (0,1,u_{x_2})=-\left\{\sqrt{|x|^2+(p-u(x))^2}\right\}_{x_2}.$
Hence ${\bf m}(x)$ satisfies \eqref{eq:curl condition on m}. 
If we set 
\begin{equation}\label{eq:definition of h(x)}
h(x)=-\sqrt{|x|^2+(p-u(x))^2},
\end{equation} 
then $h$ satisfies \eqref{eq:gradient of h in terms of m}.

Let us calculate the matrix inside the determinant in \eqref{eq:determinant condition}. 
From the definition of ${\bf m}$
\begin{align*}
{\bf e}-{\bf m}
&=-\dfrac{1}{h}\left(x,u-p-h\right)\notag\\
{\bf e}\cdot {\bf m}+\dfrac{\({\bf e}-{\bf m}\)\cdot \nu}{\sqrt{1+|\nabla u|^2}}
&=-\dfrac{p-u}{h}-\dfrac{1}{h}\dfrac{\left(x,u-p-h\right)\cdot (-\nabla u,1)}{1+|\nabla u|^2}.\label{eq:factor of hessian}
\end{align*}
Let us next calculate $D^2h$: 
{\small
\begin{align*}
\partial_{x_i}h&=- \dfrac{x_i+u_{x_i}\,(u-p)}{\sqrt{|x|^2+(u-p)^2}}=\dfrac{x_i+u_{x_i}\,(u-p)}{h}\\
\partial_{x_ix_j}h
&=h^{-1}
\(\delta_{ij}+u_{x_j}\,u_{x_i}+(u-p)\,u_{x_ix_j}\)-h^{-3}\,\(x_i+(u-p)\,u_{x_i}\)\,\(x_j+(u-p)\,u_{x_j}\).
\end{align*}
}
Then 
\[
D^2 h=-\dfrac{1}{h^3} \(x+(u-p)\,\nabla u\)\otimes \(x+(u-p)\,\nabla u\)
+
\dfrac{1}{h}\(Id+\nabla u\otimes \nabla u+(u-p)\,D^2u\).
\]
Combining these calculations it follows that the determinant in \eqref{eq:determinant condition} equals
{\small
$$
 \det\left(\dfrac{1}{h}\(Id+\nabla u\otimes \nabla u\)-\dfrac{1}{h^3} \(x+(u-p)\,\nabla u\)\otimes \(x+(u-p)\,\nabla u\)+\dfrac{1}{h}\dfrac{\(x,u-p-h\)\cdot (-\nabla u,1)}{1+|\nabla u|^2}D^2u\right),
 $$
}
and factoring out $1/h$ yields \eqref{eq:near field condition}.
\end{proof}

\begin{remark}\label{rmk:radial Gamma existence of phase}\rm
With this proposition we show existence of phases when $\Gamma$ is given by the graph of a radial function $u$.
Let us write $u(x):=v(|x|^2)$ and we have $p>u(x)$ for all $x$.
Setting $r=|x|^2=x_1^2+x_2^2$ it follows that
\begin{align*}
u_{x_i}=2x_i\,v'(r),\qquad
u_{x_ix_j}=2\,\delta_{ij}\,v'(r)+4\,x_i\,x_j\,v''(r).
\end{align*}
Hence
\begin{align*}
\nabla u=2v'(r)\, x,\qquad 
D^2u= 2v'(r)\,Id+4v''(r)\, x\otimes x.
\end{align*}
We now calculate the components of the matrix $\mathcal A$ inside the determinant in \eqref{eq:near field condition}: 
\begin{align*}
\nabla u\otimes \nabla u&=4\left(v'(r)\right)^2\,x\otimes x\\
\(x+(u-p)\nabla u\)\otimes \(x+(u-p)\nabla u\)&=\(x+2(v(r)-p)v'(r)\,x\)\otimes \(x+2(v(r)-p)v'(r)\,x\)\\
&=\left(1+2v'(r)(v(r)-p)\right)^2\,x\otimes x\\
\dfrac{(x,u-p-h)\cdot(-\nabla u,1)}{1+|\nabla u|^2}&=\dfrac{(x,v(r)-p-h)\cdot\(-2v'(r)\,x,1\)}{1+4\,\left(v'(r)\right)^2\,r}\\
&=\dfrac{-2r\,v'(r)+v(r)-p-h}{1+4\,(v'(r))^2\,r}.
\end{align*}
Notice that in this case, the function $h$ in \eqref{eq:definition of h(x)} is radial with $h=h(r)=-\sqrt{r+(p-v(r))^2}$. Replacing in the formula for $\mathcal A$ we obtain
{\small
\begin{align*}
\mathcal A&=Id+ 4\left(v'\right)^2\,x\otimes x-\dfrac{1}{h^2} \left(1+2v'(v-p)\right)^2\,x\otimes x+\dfrac{-2r\,v'+v-p-h}{1+4\,(v')^2\,r}\left(2v'\,Id+4v''\,x\otimes x.\right)\\
&=\left(1+2\,v'\,\dfrac{-2r\,v'+v-p-h}{1+4\,(v')^2\,r}\right)Id+\(4(v')^2-\dfrac{1}{h^2} \left(1+2v'(v-p)\right)^2+4\dfrac{-2r\,v'+v-p-h}{1+4\,(v')^2\,r}v''\)x\otimes x.
\end{align*}
}
Notice that if $x=0$, then $\mathcal A=Id$, and therefore from Proposition \ref{cor:near field} there is a phase discontinuity in a neighborhood of $x=0$.

\end{remark}

\setcounter{equation}{0}
\section{Analysis of chromatic dispersion in metalenses}

\subsection{Dispersion of dichromatic light}
Let $(\Gamma,\phi)$ be a metasurface in $\R^3$ surrounded by vacuum. A dichromatic ray with unit direction ${\bf x}$, superposition of two wavelengths $\lambda_1$ and $\lambda_2$, strikes $\Gamma$ at a point $P$. 
Let $ \nu(P)$ be the outer-unit normal to $\Gamma$ at $P$, and let ${\bf m}_1$, ${\bf m}_2$ be the unit directions of the refracted rays at $P$ corresponding to $\lambda_1$ and $\lambda_2$. 
From \eqref{eq:SnellLaw}, the vectors
${\bf x}-{\bf m_1}-\nabla (\lambda_1\phi(P)/2\pi)$ and ${\bf x}-{\bf m_2}-\nabla (\lambda_2\phi(P)/2\pi)$ are parallel to $\nu(P)$. We show that the following statements are equivalent
\begin{enumerate}
\item ${\bf m_1}={\bf m_2}$,
\item $\nabla \phi(P)$ is colinear to $\nu(P)$,
\item ${\bf m_1}={\bf m_2}={\bf x}$.
\end{enumerate}
This means that if $\nabla \phi(P)$ is not parallel to $\nu(P)$, then rays composed with different colors will be dispersed by the metalens, i.e., refracted in different directions. 
And only in case $\nabla \phi(P)$ is parallel to $\nu(P)$, the incident ray propagates through the metalens without changing direction.

\begin{proof}
$(3)\implies (1)$ is trivial, so we will show that 
$(1)\implies (2)\implies (3).$

\textbf{$(1)\implies (2)$.} 
Since $n_1=n_2$, from \eqref{eq:SnellLaw}
$$\left({\bf x}-{\bf m_1}\right) \times  \nu(P)=\nabla \(\lambda_1\phi(P)/2\pi\)\times \nu(P),
\qquad ({\bf x}-{\bf m_2})\times \nu(P)=\nabla \(\lambda_2\phi(P)/2\pi\)\times  \nu(P),$$
implying
$\dfrac{\lambda_1-\lambda_2}{2\pi}\nabla \phi(P)\times  \nu(P)=0,$ and since $\lambda_1\neq\lambda_2$, $(2)$ follows.
\\
\textbf{$(2)\implies (3)$.} If $\nabla \phi(P)$ is parallel to $\nu(P)$, then from \eqref{eq:SnellLaw}
$\({\bf x}-{\bf m_1}\)\times  \nu(P)=0,$
so there exists $\mu_1\in \R$ such that
\begin{equation}\label{eq:refraction special case}
{\bf x}-{\bf m_1}=\mu_1 \nu(P).
\end{equation}
Dotting \eqref{eq:refraction special case} with ${\bf x}$ and ${\bf m}_1$ yields
$$1- {\bf m_1}\cdot {\bf x}=\mu_1{\bf x}\cdot \nu(P)\qquad\qquad {\bf m_1}\cdot {\bf x}-1=\mu_1 {\bf m_1}\cdot \nu(P),$$
and adding these identities, we obtain $\mu_1 ({\bf x}+{\bf m_1})\cdot \nu(P)=0.$  
 Hence $\mu_1=0$ or ${\bf x}\cdot \nu(P)={\bf m_1}\cdot \nu(P)=0$, since ${\bf x}\cdot \nu(P)$ and ${\bf m_1}\cdot \nu(P)$ are both non negative.  Dotting \eqref{eq:refraction special case} with $\nu(P)$, if  ${\bf x}\cdot \nu(P)={\bf m_1}\cdot \nu(P)=0$ then we also get $\mu_1=0$. 
Therefore ${\bf x}={\bf m_1},$
and similarly, ${\bf x}={\bf m_2},$ concluding $(3)$.
\end{proof}

\subsection{Analysis of the chromatic aberration for a plane metasurface}\label{sec:chromatic aberration plane}
Let $\Gamma$ be the horizontal plane $x_3=a$ in $\R^3$.
Suppose a phase discontinuity $\phi$ is defined in a neighborhood of $\Gamma$ such that vertical rays having color with wavelength $\lambda_0$ and striking $\Gamma$ are refracted into a point $P_0=(0,0,p_0)$ above $\Gamma$ as in Figure \ref{fig:pic for m(x)}. 
We then ask how this planar metasurface $(\Gamma,\phi)$ focuses vertical rays with wavelength $\lambda\neq \lambda_0$. 
Assuming that all incoming rays pass through the circular aperture $x_1^2+x_2^2\leq R^2$, we shall prove 
that for appropriate values of $\lambda$ each refracted ray with wavelength $\lambda$ intersects the $z$-axis at some point $P_{\lambda}(x)=(0,0,p_{\lambda}(x))$ depending on $x=(x_1,x_2)$. In addition, we shall prove an estimate for the distance between the points $P_{\lambda}$ and $P_0$ in terms of $ p_0-a$ and the ratio $\lambda/\lambda_0$, see \eqref{eq:estimate of the distance between P and P0} and Figure \ref{fig:pic of rays with two colors}. We summarize the results of this section in the following theorem.

\begin{theorem}\label{thm:planar case error}
Let $P_0=(0,0,p_0)$. Consider the metalens $(\Gamma, \phi)$ surrounded by vacuum, with $\Gamma$ the horizontal plane $x_3=a$ in $\R^3$, $a<p_0$ and with
\begin{equation}\label{eq:phi planar}
\phi(x_1,x_2,x_3)=\dfrac{2\pi}{\lambda_0}\,\sqrt{x_1^2+x_2^2+(p_0-a)^2}+g(x_3)=
\dfrac{2\pi}{\lambda_0}\,\dist\((x,a),P_0\)+g(x_3)
\end{equation}
where $g$ is an arbitrary $C^1$ function satisfying $g'(a)=0$.
Suppose vertical rays that are superposition of two colors, with wavelengths $\lambda_0$ and $\lambda$, strike $\Gamma$ at  
$(x_1,x_2,a)$, with $x_1^2+x_2^2<R^2$. 

Then each ray splits into two rays one with wavelength $\lambda_0$ and another with wavelength $\lambda$.
Each ray with wavelength $\lambda_0$ is refracted into the point $P_0$. 
And, if 
$\dfrac{\lambda}{\lambda_0}<1$, or $1< \dfrac{\lambda}{\lambda_0}\leq \sqrt{1+\dfrac{(p_0-a)^2}{R^2}} $, then 
each ray with wavelength $\lambda$ is refracted into the point $P_{\lambda}(x)=\(0,0,\dfrac{\lambda_0}{\lambda}\sqrt{\(1-\(\lambda/\lambda_0\)^2\)\,|x|^2+(p_0-a)^2}\).$
We also have 
\begin{equation}\label{eq:estimate of the distance between P and P0}
|P_0-P_{\lambda}(x)|\leq
\left|1-\dfrac{\lambda_0}{\lambda}\right|\,\(p_0-a\)
\end{equation}
for each $x=(x_1,x_2)$ with $x_1^2+x_2^2\leq R^2$.

\end{theorem}

\begin{figure}[htp]
\begin{center}
\subfigure[Rays with wavelength $\lambda_0$ focus into a fixed point]{\label{fig:pic for m(x)}\includegraphics[width=2.5in]{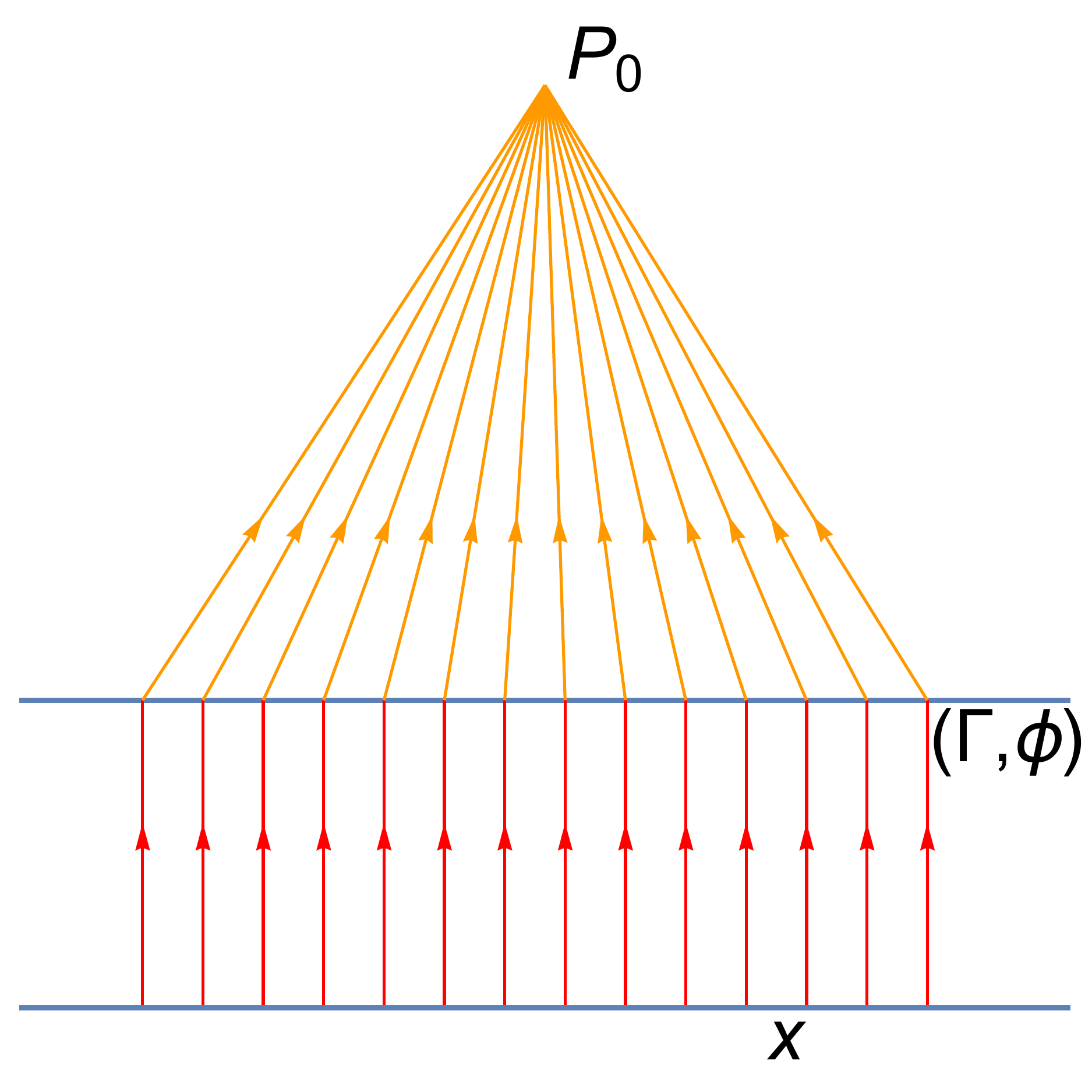}}
\subfigure[Rays with wavelength $\lambda$ are defocused]{\label{fig:pic of rays with two colors}\includegraphics[width=2.5in]{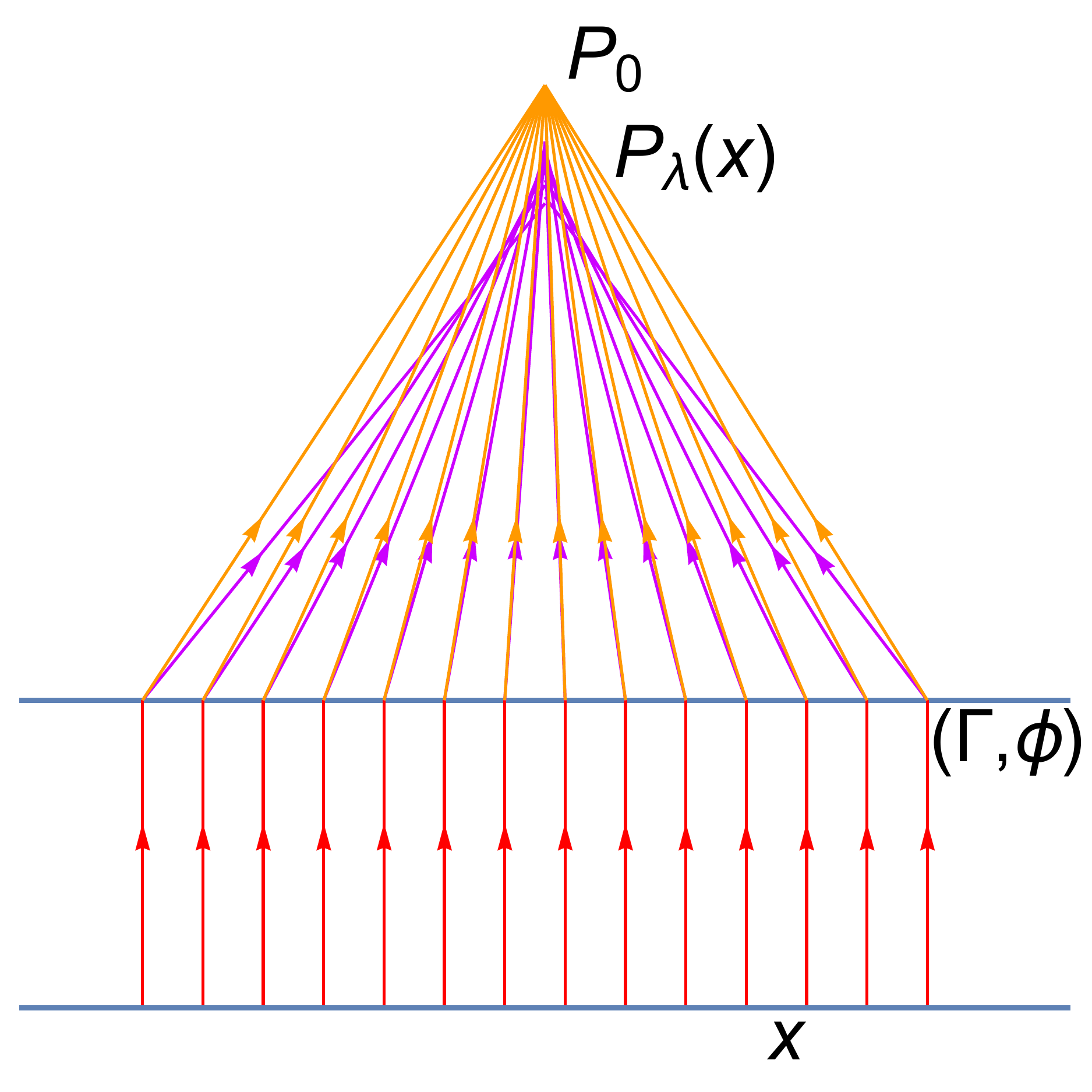}}
\end{center}
  \caption{$ $}
  \label{fig:fields}
\end{figure}

\begin{proof}
Let us first calculate the phase discontinuity focusing rays with color $\lambda_0$ into $P_0$. The incident ray  with direction ${\bf e}=(0,0,1)$ emanating from $x=(x_1,x_2)$ with $x_1^2+x_2^2\leq R^2$ strikes the plane $x_3=a$ and then bends to strike the point $P_0$. 
Applying the set up from Section \ref{sec:general field of directions} in this case we have 
${\bf m}(x)=\dfrac{P_0-(x,a)}{|P_0-(x,a)|}$ and the normal $\nu={\bf e}$. Since we seek $\phi$ tangential to $\Gamma$, we then have from \eqref{eq:nabla phi general surface} that 
\begin{equation}\label{eq:formula for V planar case}
V(x,a)=\nabla \phi(x,a)=\dfrac{2\pi}{\lambda_0}\(-{\bf m}+\({\bf m}\cdot {\bf e}\)\,{\bf e}\)
=
\dfrac{2\pi}{\lambda_0}\(\dfrac{x}{\sqrt{|x|^2+(p_0-a)^2}},0\).
\end{equation}
Clearly, the function $\phi$ in \eqref{eq:phi planar} is the desired phase function.
Therefore, the planar metasurface having this phase focuses all rays with color $\lambda_0$ into the point $P_0$.

Let us now see how the metalens $(\Gamma,\phi)$ focuses rays with wavelength $\lambda$. 
First, to avoid total reflection, from \eqref{eq:internal reflection}, rays with color $\lambda$ are refracted by the metalens $(\Gamma,\phi)$ if 
\begin{equation}\label{eq:internal reflection lambda}
\left[\left({\bf e}-\nabla\(\lambda\,\phi/2\pi\)\right)\cdot \nu\right]^2\geq \left|{\bf e}-\nabla\(\lambda\,\phi/2\pi\)\right|^2-1.
\end{equation}
Since $\nu=(0,0,1)={\bf e},$ and $\nabla \phi\cdot \nu=0$, then 
\eqref{eq:internal reflection lambda} is equivalent to
$$0\leq 2-\left|{\bf e}-\nabla \(\lambda\,\phi/2\pi\) \right|^2=\dfrac{\(1-\(\lambda/\lambda_0\)^2\)\,|x|^2+(p_0-a)^2}{|x|^2+(p_0-a)^2}:=\dfrac{\Delta(x)}{|x|^2+(p_0-a)^2}.$$
Clearly, this inequality holds for $\lambda<\lambda_0$. If $\lambda>\lambda_0$, then $\Delta\geq 0$ if and only if $(p_0-a)^2\geq \(\(\lambda/\lambda_0\)^2-1\)\,|x|^2$. Since $|x|^2\leq R^2$, if $\lambda>\lambda_0$ satisfies
\[
\sqrt{\(\lambda/\lambda_0\)^2-1}\leq \dfrac{p_0-a}{R}
\]
then $\Delta\geq 0$ for all $|x|\leq R$. Therefore for these values of $\lambda$ total reflection is avoided.

Next, we need to know the directions of the refracted rays with color $\lambda$. 
According to \eqref{eq:SnellLaw bis}, $(\Gamma,\phi)$ refracts rays with color $\lambda$ into the unit direction ${\bf m'}$ with 
\begin{equation}\label{eq:refraction with color lambda}
{\bf m'}= {\bf e}-\mu'\,\nu -\nabla \(\lambda\,\phi/2\pi\),
\end{equation}
where from \eqref{eq:formula for mu}
\[
\mu'=1-\sqrt{2-\left|{\bf e}-\nabla \(\lambda\,\phi/2\pi\) \right|^2}=1-
\dfrac{1}{\sqrt{|x|^2+(p_0-a)^2}}
\,
\sqrt{\Delta(x)}.
\]
Writing \eqref{eq:refraction with color lambda} in coordinates yields
\begin{align*}
m_i'&=-\lambda\,\phi_{x_i}/2\pi=-\dfrac{\lambda}{\lambda_0}\,\dfrac{x_i}{\sqrt{|x|^2+(p_0-a)^2}},\qquad i=1,2\\
m_3'&=1-\mu=\dfrac{1}{\sqrt{|x|^2+(p_0-a)^2}}\,
\sqrt{\Delta(x)}.
\end{align*}

We now see where the line $t\,{\bf m}'+(x,a)$ intersects the vertical line $x_1=x_2=0$, that is, we find the point on this line that is focused by the refracted ray emanating from $x=(x_1,x_2)$ with direction ${\bf e}$ and color $\lambda$.
So we need to find $t$ so that $t\,m_1'+x_1=0$ and $t\,m_2'+x_2=0$. 
This means 
$t= \dfrac{\lambda_0}{\lambda}\,\sqrt{|x|^2+(p_0-a)^2}$.
Hence 
$
t\,m_3'+a=\dfrac{\lambda_0}{\lambda}\,\sqrt{\Delta}+a
$, 
and the focused point on the vertical line by the vertical ray emanating from $x$ with color $\lambda$ is $P_\lambda(x)=\(0,0,\dfrac{\lambda_0}{\lambda}\,\sqrt{\Delta}+a\)$.

We now calculate the distance between $P_0$ and $P_\lambda(x)$.
First write
\begin{align*}
t\,m_3'+a-p_0&=\dfrac{\lambda_0}{\lambda}\,\sqrt{\Delta}-\(p_0-a\)
=
\dfrac{\(\dfrac{\lambda_0}{\lambda}\,\sqrt{\Delta}-\(p_0-a\)\)\(\dfrac{\lambda_0}{\lambda}\,\sqrt{\Delta}+\(p_0-a\)\)}{\dfrac{\lambda_0}{\lambda}\,\sqrt{\Delta}+\(p_0-a\)}\\
&=
\dfrac{\(\dfrac{\lambda_0}{\lambda}\)^2\,\Delta-\(p_0-a\)^2}{\dfrac{\lambda_0}{\lambda}\,\sqrt{\Delta}+\(p_0-a\)}
=
\dfrac{\(\dfrac{\lambda_0}{\lambda}\)^2\(1-\(\dfrac{\lambda}{\lambda_0}\)^2\)\,|x|^2
+\(\(\dfrac{\lambda_0}{\lambda}\)^2-1\)\(p_0-a\)^2}{\dfrac{\lambda_0}{\lambda}\,\sqrt{\Delta}+\(p_0-a\)}\\
&=
\(\(\dfrac{\lambda_0}{\lambda}\)^2-1\)
\,\dfrac{|x|^2
+\(p_0-a\)^2}{\dfrac{\lambda_0}{\lambda}\,\sqrt{\Delta}+\(p_0-a\)}.
\end{align*}
We then obtain
\begin{equation}\label{eq:error formula for plane}
\dist\(P_0,P_\lambda(x)\)=\left| t\,m_3'+a-p_0\right|
= 
\left|\(\dfrac{\lambda_0}{\lambda}\)^2-1\right|
\,\dfrac{|x|^2
+\(p_0-a\)^2}{\dfrac{\lambda_0}{\lambda}\,\sqrt{\Delta}+\(p_0-a\)},
\end{equation}
since the right hand side is a radial function of $x$, then rays with color $\lambda$ emitted from all points in a circle are focused into the same point. 
Set $|x|^2=s$ and consider the function $F(s)=\dfrac{s
+\(p_0-a\)^2}{\dfrac{\lambda_0}{\lambda}\,\sqrt{\(1-\(\lambda/\lambda_0\)^2\)\,s+(p_0-a)^2}+\(p_0-a\)}$. 
To estimate the error \eqref{eq:error formula for plane} and obtain \eqref{eq:estimate of the distance between P and P0}, we will find the maximum of $F(s)$ for $s\in [0,R^2]$. 
Suppose first that $\lambda<\lambda_0$. In this case $\Delta>0$ and it is easy to see that $F'(s)>0$ in the interval $[0,R^2]$ and so $F$ is increasing. 
Therefore when $\lambda<\lambda_0$
\[
\left| t\,m_3'+a-p_0\right|\leq \left|\(\dfrac{\lambda_0}{\lambda}\)^2-1\right|\,F(0)=
\(\dfrac{\lambda_0}{\lambda}-1\)\,\(p_0-a\).
\]
On the other hand, if $\lambda_0\sqrt{1+\dfrac{(p_0-a)^2}{R^2}} >\lambda>\lambda_0$, then it is also easy to see that $F$ is increasing in $[0,R^2]$ and so the error 
\[
\left| t\,m_3'+a-p_0\right|\leq \left|\(\dfrac{\lambda_0}{\lambda}\)^2-1\right|\,F(0)=
\(1-\dfrac{\lambda_0}{\lambda}\)\,\(p_0-a\).
\]
Therefore, we conclude the estimate \eqref{eq:estimate of the distance between P and P0}.
\end{proof}
\subsection{Analysis of the chromatic aberration for a general radial function $u$}\label{subsec:aberration for general radial u}
In this section, the surface $\Gamma$ is given by the graph of a function $u(x)=v(|x|^2)$. 
In view of Remark \ref{rmk:radial Gamma existence of phase}, a phase $\phi$ exists in a neighborhood of $x=0$ so that the metalens $(\Gamma,\phi)$ refracts all vertical rays with wavelength $\lambda_0$ into a point $P_0=(0,0,p_0)$. 
We analyze chromatic aberration caused by $(\Gamma,\phi)$, that is, how rays with wavelength $\lambda\neq \lambda_0$ are focused. We shall prove that the order of magnitude of the focusing error in the radial case is as in the planar case \eqref{eq:estimate of the distance between P and P0}. 

In order to do this, we first find an expression for $\nabla \phi(x,u(x))$ when $u(x)=v(|x|^2).$ 
In this case, the set of transmitted directions for rays with wavelength $\lambda_0$ is
\[
{\bf m}(x)=\dfrac{(-x,p_0-u(x))}{\sqrt{|x|^2+(p_0-u(x))^2}}.
\]
And from Theorem \ref{thm:metalens  existence} the gradient of $\phi$ satisfies \eqref{eq:nabla phi general surface}.
As in Remark \ref{rmk:radial Gamma existence of phase}, we let $h(x)=-\sqrt{|x|^2+(p_0-u(x))^2}$. 
Set $r=|x|^2,$ so $\nabla u= 2\,v'(r)\,x$, and $h(r)=-\sqrt{r+(p_0-v(r))^2}$. Then
\begin{align*}
\({\bf e}-{\bf m}\)\cdot \nu&=-\dfrac{(x,u-p_0-h)}{h}\cdot\dfrac{(-\nabla u,1)}{\sqrt{1+|\nabla u|^2}}\\
&=\dfrac{-1}{h\,\sqrt{1+|\nabla u|^2}}\(-x\cdot \nabla u + u -p_0-h\)\\
&=\dfrac{-1}{h\,\sqrt{1+4\,v'(r)^2\,r}}\left(-2\,v'(r)\,r+v(r)-p_0-h\right)
\end{align*}
so from \eqref{eq:nabla phi general surface}
\begin{align}\label{eq:gradient phi radial}
\nabla \phi(x,u(x))&=\dfrac{2\pi}{\lambda_0}
\(-\dfrac{(x,v(r)-p_0-h)}{h}+\dfrac{1}{h\,\sqrt{1+4\,v'(r)^2\,r}}\left(-2\,v'(r)\,r+v(r)-p_0-h\right)\,\dfrac{\(-2v'(r)x ,1\)}{\sqrt{1+4v'(r)^2r}}\)\\\notag
&=
-\dfrac{2\pi}{\lambda_0\,h}\left(\dfrac{1+2v'(r)(v(r)-p_0-h)}{1+4v'(r)^2r}\right)\left(x,2v'(r)\,r\right).
\end{align}
To simplify the notation set $A(r)=1+2v'(r)(v(r)-p_0-h)$.

Since the phase $\phi$ is defined in a neighborhood of $x=0$, we assume that $|x|^2=r<R^2$ for some $R$ small.

As in Section \ref{sec:chromatic aberration plane}, we next find conditions on $\lambda$, $\lambda_0$ and $R$ so that rays with wavelength $\lambda$ are not totally internally reflected. We have $\nabla \phi\cdot \nu=0$, then from \eqref{eq:internal reflection} with $n_1=n_2=1$, ${\bf x}={\bf e}$, to avoid total reflection we should verify that
\begin{equation}\label{eq:condition radial internal reflection}
\left({\bf e}\cdot \nu\right)^2\geq \left|{\bf e}-\nabla\left(\lambda\phi/2\pi\right)\right|^2-1.
\end{equation}
Indeed, we have $\left({\bf e}\cdot \nu\right)^2=\dfrac{1}{1+4v'(r)^2r}$, and from \eqref{eq:gradient phi radial}
{\small
\begin{align*}
 \left|{\bf e}-\nabla\left(\lambda\phi/2\pi\right)\right|^2-1&=-2\,{\bf e}\cdot\nabla\left(\lambda\phi/2\pi\right)+\left|\nabla\left(\lambda\phi/2\pi\right)\right|^2\\
&=-2\left(-\dfrac{\lambda}{\lambda_0\,h}\left(\dfrac{A(r)}{1+4v'(r)^2r}\right)2v'(r)r\right)+\left(\dfrac{\lambda}{\lambda_0\,h}\right)^2\dfrac{A(r)^2}{1+4v'(r)^2r}r\\
&=\dfrac{\lambda}{\lambda_0h}\left(\dfrac{A(r)}{1+4v'(r)^2r}\right)\left(\dfrac{\lambda}{\lambda_0\,h}
A(r)+4v'(r)\right)r.
 \end{align*}
}
Hence \eqref{eq:condition radial internal reflection} is satisfied if and only if
$$1\geq 
\dfrac{\lambda}{\lambda_0 h}A(r)\,\left(\dfrac{\lambda}{\lambda_0\,h}\,A(r)+4v'(r)\right)r.$$
We conclude that taking $R$ small enough, above inequality holds for $r\leq R^2$, and total internal reflection is avoided for rays with color $\lambda$.

We next study the direction of the refracted rays with wavelength $\lambda$. From \eqref{eq:SnellLaw bis} the direction of the refracted ray with color $\lambda$ is
\[
{\bf m'}={\bf e}-\mu'\,\dfrac{\(-2v'(r)x,1\)}{\sqrt{1+4v'(r)^2r}}
-
\dfrac{\lambda}{2\pi}\nabla \phi,
\]
where from \eqref{eq:formula for mu} and the above calculation
{\footnotesize
\begin{align}\label{eq:formula for mu'}
\mu'&={\bf e}\cdot \nu-\sqrt{1-\left(\left|{\bf e}-\nabla \(\lambda\,\phi/2\pi\)\right|^2-\left({\bf e}\cdot \nu\right)^2\right)}\\\notag
&=\dfrac{1}{\sqrt{1+4v'(r)^2r}}-
\sqrt{-\dfrac{\lambda}{\lambda_0h}\left(\dfrac{A(r)}{1+4v'(r)^2r}\right)\left(\dfrac{\lambda}{\lambda_0\,h}
A(r)+4v'(r)\right)r+\dfrac{1}{1+4v'(r)^2r}}\\\notag
&=\dfrac{1}{\sqrt{1+4v'(r)^2r}}
\left(1-\sqrt{1-\dfrac{\lambda}{\lambda_0h}A(r)\,\left(\dfrac{\lambda}{\lambda_0\,h}
A(r)+4v'(r)\right)r}\right).
\end{align}
}

Writing ${\bf m'}=\(m_1',m_2',m_3'\)$ yields
\begin{align*}
\(m_1',m_2'\)&=\mu'\,\dfrac{2v'(r)x}{\sqrt{1+4v'(r)^2r}}
+
\dfrac{\lambda}{\lambda_0\,h}\left(\dfrac{A(r)}{1+4v'(r)^2r}\right)x\\
&=\dfrac{1}{1+4v'(r)^2r}\left[2v'(r)\left(1-\sqrt{1-\dfrac{\lambda}{\lambda_0h}A(r)\left(\dfrac{\lambda}{\lambda_0\,h}A(r)+4v'(r)\right)r}\right)+\dfrac{\lambda}{\lambda_0h}A(r)\right]x.
\end{align*}
Recall that $r$ is chosen small enough so that the term inside the square root in the identity above is positive.

Now we find $t$ such that $t\,\(m_1',m_2'\)+x=0$, that is,
\[
\dfrac{t}{1+4v'(r)^2r}\left[2v'(r)\left(1-\sqrt{1-\dfrac{\lambda}{\lambda_0h}A(r)\left(\dfrac{\lambda}{\lambda_0\,h}A(r)+4v'(r)\right)r}\right)+\dfrac{\lambda}{\lambda_0h}A(r)\right]+1=0
\]
obtaining
\[
t=-\dfrac{1+4v'(r)^2r}{2v'(r)\left(1-\sqrt{1-\dfrac{\lambda}{\lambda_0h}A(r)\left(\dfrac{\lambda}{\lambda_0\,h}A(r)+4v'(r)\right)r}\right)+\dfrac{\lambda}{\lambda_0h}A(r)}.\\
\]

Next,
\begin{align*}
m_3'&=1-\mu'\,\dfrac{1}{\sqrt{1+4v'(r)^2r}}-\dfrac{\lambda}{2\pi}\phi_{x_3}(x,u(x))\\
&=1-\dfrac{1}{1+4v'(r)^2r}\left(1-\sqrt{1-\dfrac{\lambda}{\lambda_0h}A(r)\left(\dfrac{\lambda}{\lambda_0\,h}A(r)+4v'(r)\right)r}\right)+\dfrac{\lambda}{\lambda_0h}\left(\dfrac{A(r)}{1+4v'(r)^2r}\right)2v'(r)r\\
&=\dfrac{1}{1+4v'(r)^2r}\left(4v'(r)^2r+\sqrt{1-\dfrac{\lambda}{\lambda_0h}A(r)\left(\dfrac{\lambda}{\lambda_0\,h}A(r)+4v'(r)\right)r}+\dfrac{\lambda}{\lambda_0h}A(r)(2v'(r)r)\right).
\end{align*}
The ray with color $\lambda$ then focuses at the point $P_\lambda=\(0,0,t\,m_3'+u(x)\)$, and we want to see how far is this point from $P_0=(0,0,p_0)$. So 
we then need to estimate the error
\[
t\,m_3'+u(x)-p_0.
\]
Taking limits when $r\to 0$ yields
\begin{align*}
v(r)&\to u(0),\quad 
h(r)\to u(0)-p_0,\quad
A(r)\to 1,\quad
t\to \frac{\lambda_0}{\lambda}(p_0-u(0)),\quad
m_3'\to 1.
\end{align*}
Therefore 
\begin{equation}\label{eq:order of magnitude for radial case}
|t\,m_3'+u(x)-p_0|\to 
\left|\dfrac{\lambda_0}{\lambda}-1\right|\(p_0-u(0)\),\quad \text{as $r\to 0$,}
\end{equation}
obtaining in the radial case that the order of magnitude of the error is as in the planar case \eqref{eq:estimate of the distance between P and P0}.

\section{Comparison with chromatic aberration in standard lenses}\label{sec:standard lens}
We analyze here the chromatic aberration in a standard lens sandwiched by a horizontal plane and a hyperboloid, and compare this dispersion with the one obtained for metasurfaces.

It is known that hyperboloids having appropriate eccentricity refract vertical rays into their focus point. More precisely, for a fixed wavelength $\lambda$,
we have two materials $I,II$ with corresponding refractive indices $n_1,n_2$, respectively,  
and a point $Y=(y,y_3)\in \R^3$ located in material $II$ to be focused; let $\kappa=n_1/n_2>1$. 
Let 
\[
h(x)=y_{n+1}-\dfrac{\kappa\,b}{\kappa^2-1}-\sqrt{\(\dfrac{b}{\kappa^2-1}\)^2+\dfrac{|x-y|^2}{\kappa^2-1}},\qquad b>0,
\]
whose graph is a sheet of a hyperboloid with upper focus $Y$ and eccentricity $\kappa$ as shown in Figure \ref{fig:hyperboloid}.
\begin{figure}[htp]
\begin{center}
\includegraphics[width=2.5in]{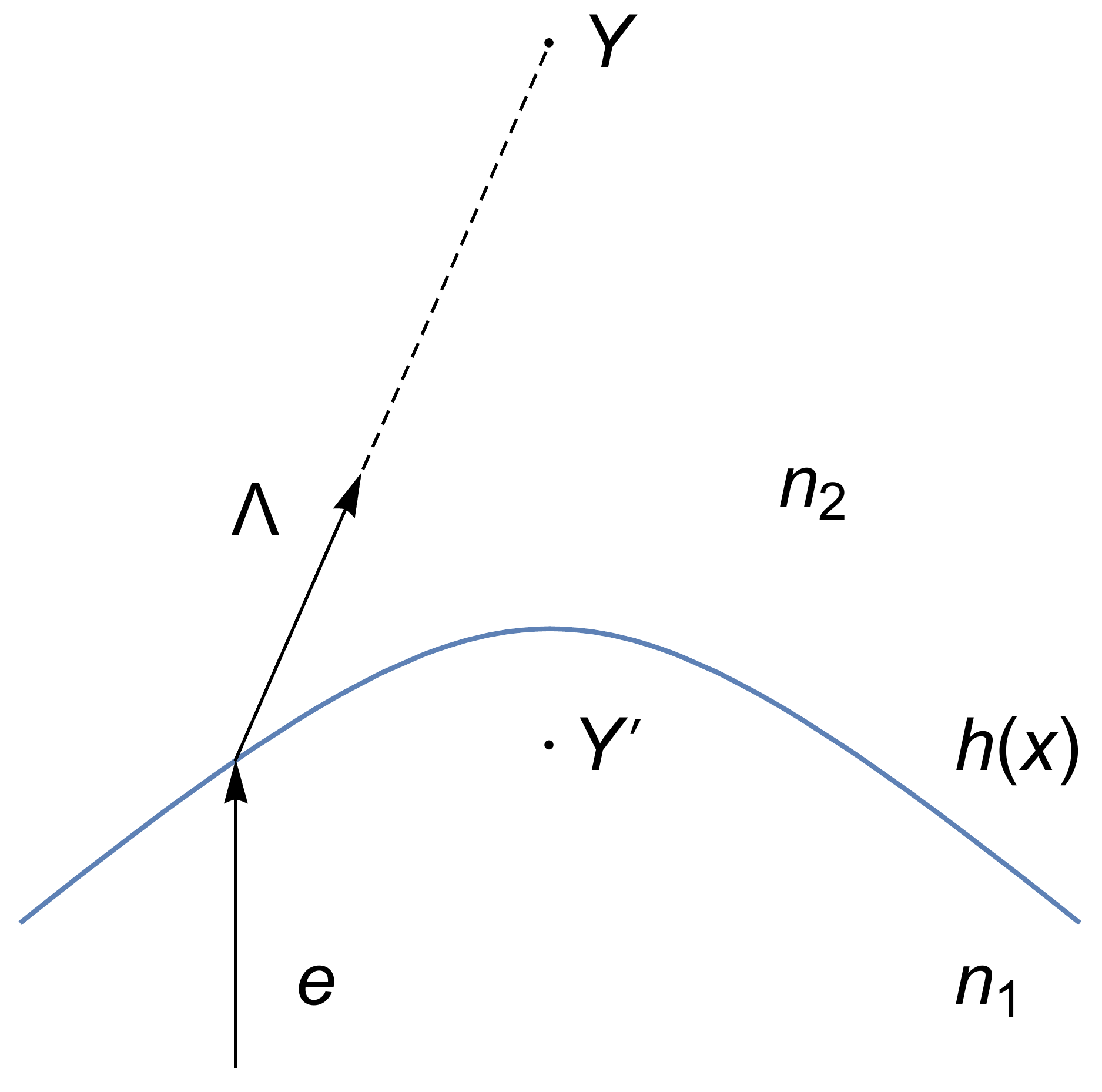}
\end{center}
  \caption{$ $}
  \label{fig:hyperboloid}
\end{figure}
Suppose the material below this hyperboloid is $I$ and the material above is $II$.
If a vertical ray with wavelength $\lambda$ strikes the hyperboloid at a point $(x,h(x))$, then this ray is refracted into a direction passing through the point $Y$ as shown in Figure \ref{fig:hyperboloid}.
We analyze how a vertical ray having different wavelength would be refracted by the same hyperboloid. In other words, how far is the new refracted point from the focus $Y$?   

We assume that $II$ is vacuum, so $n_2=1$ for all wavelengths, $n_1=n=n(\lambda)$ for the wavelength $\lambda$, and the focus $Y=(0,0,0)$;
 ${\bf e}=(0,0,1)$. From the standard Snell law, the incoming ray from below the hyperboloid with direction ${\bf e}$ is refracted by the hyperboloid into a ray with unit direction 
\[
\Lambda=\kappa\,{\bf e}+\delta\,\nu
\]
where $\nu$ is the outer unit normal to $h$ at the striking point, and 
\[
\delta=-\kappa\,({\bf e}\cdot \nu)+\sqrt{1-\kappa^2+\kappa^2\,({\bf e}\cdot \nu)^2},
\]
see \cite[Section 2]{abedin-gutierrez-tralli:regularity-refractors}; with $\kappa=n$.
Suppose now that the vertical ray has wavelength $\lambda'$ so the refractive index for the material $I$ under this wavelength has value $n'=n(\lambda')$. From the Snell law such a ray is then refracted by the hyperboloid into a direction
\[
\Lambda'=\kappa'\,{\bf e}+\delta'\,\nu
\]
with 
\[
\delta'=-\kappa'\,({\bf e}\cdot \nu)+\sqrt{1-\kappa'^2+\kappa'^2\,({\bf e}\cdot \nu)^2},
\]
where $\kappa'=n'$. 

Consider the line through the point $X=(x,h(x))$ with direction $\Lambda'$, and we want to determine where this line intersects the vertical line $x=0$.
That is, we need to find $t$ such that $X+t\,\Lambda'$ intersects the vertical line $x=0$,
and see how far this intersection point is from the original focus $(0,0,0)$. 
To do this we have
$
\nu=\dfrac{\(-\nabla h,1\)}{\sqrt{1+|\nabla h|^2}},
$
so
\[
\Lambda'=n'\,(0,0,1)+\delta'\,\nu=
\(\delta' \(-\dfrac{\nabla h}{\sqrt{1+|\nabla h|^2}}\),n'+\delta'\,\dfrac{1}{\sqrt{1+|\nabla h|^2}}\)
\]
and therefore $t$ must satisfy 
$
x+t\,\delta' \(-\dfrac{\nabla h}{\sqrt{1+|\nabla h|^2}}\)=(0,0)$.
Since 
$\nabla h=\dfrac{-1}{\sqrt{b^2+(n^2-1)\,|x|^2}}\,x$, we get 
$
\sqrt{1+|\nabla h|^2}=
\dfrac{\sqrt{b^2+n^2\,|x|^2}}{\sqrt{b^2+(n^2-1)\,|x|^2}}$.
Therefore 
\[
t
=-\dfrac{\sqrt{b^2+n^2\,|x|^2}}{\delta'}.
\]
To calculate $\delta'$, we first have
$
{\bf e}\cdot \nu
=
\dfrac{\sqrt{b^2+(n^2-1)\,|x|^2}}{\sqrt{b^2+n^2\,|x|^2}},
$
hence
\begin{align*}
\delta'
&=
-n'\,\dfrac{\sqrt{b^2+(n^2-1)\,|x|^2}}{\sqrt{b^2+n^2\,|x|^2}}
+
\sqrt{1-n'^2+n'^2\(\dfrac{b^2+(n^2-1)\,|x|^2}{b^2+n^2\,|x|^2}\)}\\
&=
\dfrac{-n'\,\sqrt{b^2+(n^2-1)\,|x|^2} + \sqrt{b^2+\(n^2-n'^2\)\,|x|^2}}{\sqrt{b^2+n^2\,|x|^2}}.
\end{align*}
Thus,
\begin{align*}
t&=-\dfrac{b^2+n^2\,|x|^2}{-n'\,\sqrt{b^2+(n^2-1)\,|x|^2} + \sqrt{b^2+\(n^2-n'^2\)\,|x|^2}}\\
&=-\dfrac{\(b^2+n^2\,|x|^2\)\(n'\,\sqrt{b^2+(n^2-1)\,|x|^2} + \sqrt{b^2+\(n^2-n'^2\)\,|x|^2}\)}{b^2+\(n^2-n'^2\)\,|x|^2 -n'^2\,\(b^2+(n^2-1)\,|x|^2\)}\\
&=
\dfrac{n'\,\sqrt{b^2+(n^2-1)\,|x|^2} + \sqrt{b^2+\(n^2-n'^2\)\,|x|^2}}{n'^2-1}.
\end{align*}
The last component of $X+t\,\Lambda'$ equals
\begin{align*}
E(x)&=h(x)+t\,\(n'+\dfrac{\delta'}{\sqrt{1+|\nabla h|^2}}\)\\
&=
-\dfrac{n\,b}{n^2-1}-\sqrt{\(\dfrac{b}{n^2-1}\)^2+\dfrac{|x|^2}{n^2-1}}
+n'\,\(\dfrac{n'\,\sqrt{b^2+(n^2-1)\,|x|^2} + \sqrt{b^2+\(n^2-n'^2\)\,|x|^2}}{n'^2-1}\)\\
&\qquad \qquad 
-\sqrt{b^2+(n^2-1)\,|x|^2}\\
&=-\dfrac{n\,b}{n^2-1}
-
\dfrac{1}{n^2-1}\sqrt{b^2+\(n^2-1\)\,|x|^2}
+\dfrac{1}{n'^2-1}\,\sqrt{b^2+\(n^2-1\)\,|x|^2}\\
&\qquad \qquad+\dfrac{n'}{n'^2-1}\,\sqrt{b^2+\(n^2-n'^2\)\,|x|^2}\\
&=-\dfrac{n\,b}{n^2-1}
+\(\dfrac{1}{n'^2-1}-\dfrac{1}{n^2-1}\)\,\sqrt{b^2+\(n^2-1\)\,|x|^2}
+\dfrac{n'}{n'^2-1}\,\sqrt{b^2+\(n^2-n'^2\)\,|x|^2}.
\end{align*}
The ray with wavelength $\lambda'$ is refracted at $X$ into the point $(0,0,E(x))$.
Let
\[
g(r)=-\dfrac{n\,b}{n^2-1}
+\(\dfrac{1}{n'^2-1}-\dfrac{1}{n^2-1}\)\,\sqrt{b^2+\(n^2-1\)\,r}
+\dfrac{n'}{n'^2-1}\,\sqrt{b^2+\(n^2-n'^2\)\,r},
\]
we assume $r>0$ satisfying $b^2+\(n^2-n'^2\)\,r\geq 0$ to avoid total reflection of rays with color $\lambda'$.
If $n>n'$, then $g$ is strictly increasing and we obtain
\[
E(x)\geq E(0),
\]
for all $x$.
And if $n<n'$, then $g$ is strictly decreasing and so 
\[
E(x)\leq E(0),
\]
for $x$ satisfying $b^2+\(n^2-n'^2\)\,|x|^2\geq 0$.

We have
\[
E(0)= b\(\dfrac{1}{n'-1}-\dfrac{1}{n-1}\)
=
b\,\dfrac{n-n'}{(n'-1)(n-1)}.
\]
Let us translate the refractive indices in terms of wavelengths.
To see the order of magnitude in the error $E$ when $n$ and $n'$ are given in terms of wavelengths, we use Cauchy approximate dispersion formula, see \cite[Sec. 2.3.4, Formula (43)]{book:born-wolf} or \cite[Sec. 23.3]{jenkins-white:fundamentalsofoptics} (valid only in the visible spectrum)
\[
n(\lambda)=1+A+\dfrac{B}{\lambda^2}.
\] 
Here the terms with powers of $\lambda$ bigger that four in \cite[Sec. 2.3.4, Formula (41)]{book:born-wolf} have been neglected and $n^2$ is replaced by $2(n-1)$ when $n$ takes values for various gases; see discussion in \cite[Sec. 2.3.4, page 100]{book:born-wolf} and Table 2.6 therein. 
We set $n=n(\lambda)$ and $n'=n(\lambda')$, so from Cauchy's formula
\begin{align*}
E(0)&=b\(\dfrac{\dfrac{B}{\lambda^2}-\dfrac{B}{\lambda'^2}}{\(A+\dfrac{B}{\lambda^2}\)\(A+\dfrac{B}{\lambda'^2}\)}\)
=
b\,\dfrac{B}{\lambda^2\,\(A+\dfrac{B}{\lambda^2}\)\(A+\dfrac{B}{\lambda'^2}\)}
\(1-\(\dfrac{\lambda}{\lambda'}\)^2\)\\
&=
C(A,B,b,\lambda,\lambda')\,\(1+\dfrac{\lambda}{\lambda'}\)
\(1-\dfrac{\lambda}{\lambda'}\).
\end{align*}
The order of magnitude of this error, except for a bounded multiplicative constant, and in terms of $\lambda/\lambda'$, is similar to the order of magnitude in 
 \eqref{eq:estimate of the distance between P and P0} where $\lambda_0$ plays the role of $\lambda$ and $\lambda$ the role of $\lambda'$.
Notice that in the formula above for $E(0)$, the coefficient $b$ can be chosen arbitrarily, in particular, if $b$ is sufficiently small we can control the size of the multiplicative factor in front of $1-\dfrac{\lambda}{\lambda'}$. Notice also that using more terms in the full Cauchy dispersion formula $n(\lambda)=A+B/\lambda^2+C/\lambda^4 +D/\lambda^6+\cdots $ yields the same order of magnitude in the error $E(0)$.

\section*{Acknowledgements}{\small C. E. G. was partially supported by NSF grant DMS--1600578, and A. S. was partially supported by Research Grant 2015/19/P/ST1/02618 from the National Science Centre, Poland, entitled "Variational Problems in Optical Engineering and Free Material Design".}\\ 
\vspace{-.6cm}
\begin{wrapfigure}{l}{0.2\textwidth} 
\begin{center} \includegraphics[width=0.2\textwidth]{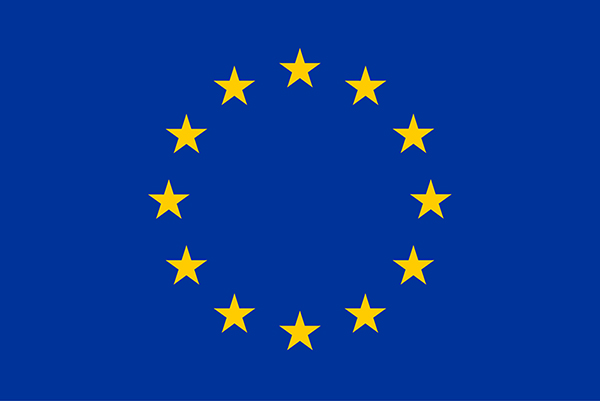} 
\end{center} 
\end{wrapfigure}\\
{\footnotesize This project has received funding from the European Union's Horizon 2020 research and innovation program under the Marie Sk\l{}odowska-Curie grant agreement No. 665778.}\\ \\

\newcommand{\etalchar}[1]{$^{#1}$}


\begin{thebibliography}{AKG{\etalchar{+}}12}

\bibitem[AGT16]{abedin-gutierrez-tralli:regularity-refractors}
F.~Abedin, C.~E. Guti\'errez, and G.~Tralli.
\newblock ${C}^{1,\alpha}$ estimates for the parallel refractor.
\newblock {\em Nonlinear Analysis: Theory, Methods \& Applications}, 142:1--25,
  2016.

\bibitem[AKG{\etalchar{+}}12]{aieta-capasso:generalizedrefractionfermat}
F.~Aieta, A.~Kabiri, P.~Genevet, N.~Yu, M.~A. Kats, Z.~Gaburro, and F.~Capasso.
\newblock Reflection and refraction of light from metasurfaces with phase
  discontinuities.
\newblock {\url{https://arxiv.org/pdf/1901.05042.pdf}},
  2019.
  
\bibitem[BMS{\etalchar{+}}19]{banerji---Menon}
S. Banerji, M. Meem, B. Sensale-Rodriguez, and R. Menon.
\newblock Imaging with flat optics: metalenses or diffractive lenses?
\newblock {\em Journal of Nanophotonics, \url{https://doi.org/10.1117/1.JNP.6.063532}}, 6(1),
  2019.
  
  
  
\bibitem[BGN{\etalchar{+}}18]{biswas-gutierrez-low-2018}
S. R. Biswas, C. E. Guti\'errez, A. Nemilentsau, In-Ho Lee, Sang-Hyun Oh, P. Avouris and Tony Low.
\newblock Tunable Graphene Metasurface Reflectarray for Cloaking, Illusion, and Focusing.
\newblock {\em Physical Review Applied, \url{https://doi.org/10.1103/PhysRevApplied.9.034021}, 9,
  2018.}

\bibitem[BW59]{book:born-wolf}
M.~Born and E.~Wolf.
\newblock {\em Principles of Optics, Electromagnetic theory, propagation,
  interference and diffraction of light}.
\newblock Cambridge University Press, seventh (expanded), 2006 edition, 1959.

\bibitem[CTY16]{Chen-Yu:Review}
H.~T. Chen, A.~J. Taylor, and N.~Yu.
\newblock A review of metasurfaces: physics and applications.
\newblock {\em Reports on  progress in physics,} 79(7), 076401,
  2016.
  
\bibitem[CZS{\etalchar{+}}18]{WangTsai:Broadband}
W.~T. Chen, A.~Y. Zhu, V.~Sanjeev, M.~Khorasaninejad, Z.~Shi, E.~Lee, and F. Capasso.
\newblock A broadband achromatic metalens for focusing and imaging in the visible.
\newblock {\em Nature nanotechnology}, 13(3): 220--226, 2018.  
  
\bibitem[F]{feynman-lectures-on-physics}
R.~Feynman.
\newblock {\em The Feynman Lectures on Physics.}
\newblock { \url{http://www.feynmanlectures.caltech.edu/}, 2010.}

\bibitem[GH09]{gutierrez-huang:reshaping light beam}
C.~E. Guti\'errez, and  G.Huang.
\newblock The refractor problem in reshaping light beams.
\newblock {\em Archive for rational mechanics and analysis}, 193(2):423--443,
  2009.

\bibitem[GPS17]{gutierrez-pallucchini-stachura:nonflatmetasurfaces}
C.~E. Guti\'errez, L.~Pallucchini, and E.~Stachura.
\newblock General refraction problems with phase discontinuities on nonflat
  metasurfaces.
\newblock {\em Journal of the Optical Society of America A}, 34(7):1160--1162,
  July 2017.

\bibitem[GP18]{gutierrez-pallucchini:metasurfacesandMAequations}
C.~E. Guti\'errez and L.~Pallucchini.
\newblock Reflection and refraction problems for metasurfaces related to {M}onge-{A}mp\`ere equations.
\newblock {\em Journal of the Optical Society of America A}, 35(9):1523--1531,
  September 2018.

\bibitem[GS16]{gutierrez-sabra:asphericallensdesignandimaging}
C.~E. Guti\'errez and and A.~Sabra.
\newblock Aspherical lens design and imaging.
\newblock {\em SIAM J. Imaging Sci.}, 9(1):386--411,
  2016.




\bibitem[GS18]{gutierrez-sabra:freeformgeneralfields}
C.~E. Guti\'errez and and A.~Sabra.
\newblock Freeform lens design for scattering data with general radiant fields.
\newblock {\em Arch. Rational Mech. Anal.}, 228:341--399,
  2018.





  
\bibitem[GRB{\etalchar{+}}18]{groever-capasso:substrate aberration}
B.~Groever, C.~Roques-Carmes, S.~Byrnes, and F.~Capasso.
\newblock Substrate aberration and correction for meta-lens imaging: an analytical approach.
\newblock {\em Applied optics,} 57(12):2983--2980,
  2018.
  
\bibitem[H]{Hecht-optics-book}
E.~Hecht.
\newblock {\em Optics.}
\newblock {Addison-Wesley, 4th Edition, 2002.}  

  
\bibitem[JW01]{jenkins-white:fundamentalsofoptics}
F.~A. Jenkins and H.~E. White.
\newblock {\em Fundamental of Optics}.
\newblock McGraw-Hill, 4th edition, 2001.

\bibitem[LC18]{lalanne-chavel}
Lalanne, Philippe and Chavel, Pierre.
\newblock Metalenses at visible wavelengths: past, present, perspectives.
\newblock {\em Laser \& Photonics Reviews,} 11(3):1600295, \url{https://doi.org/10.1002/lpor.201600295},  2017.

\bibitem[LSW{\etalchar{+}}19]{metalenses-cameras}
R. J. Lin, V.-C. Su, S. Wang, M. K. Chen, T. L. Chung, Y. H. Chen, 
H. Y. Kuo, J-W. Chen, J. Chen, Y.-T. Huang, J.-H. Wang, C. H. Chu,
P. C. Wu, T. Li, Z. Wang, S. Zhu, and D. P. Tsai.
\newblock Achromatic metalens array for full-colour light-field imaging.
\newblock {\em Nature Nanotechnology,} 14:227--231, \url{https://doi.org/10.1038/s41565-018-0347-0},  2019.


\bibitem[S]{science-runner-ups-2016}
\newblock The runners-up.
\newblock {\em Science,} 
\newblock 6319(354), 1518--1523, 2016.


\bibitem[YC14]{yu-capasso:flat optics}
N.~Yu, and F.~Capasso.
\newblock Flat optics with designer metasurfaces.
\newblock {\em Nature materials,} 13(2),
  2014.
  
\bibitem[ZKL{\etalchar{+}}17]{zhu-kutznetsov-engheta}
A. Y. Zhu, A. I. Kuznetsov, B. Luk'yanchuk, N. Engheta, and P.
Genevet.
\newblock Traditional and emerging materials for optical metasurfaces.
\newblock {\em Nanophotonics,} 6 (2):452--471, 2017.  


\end{thebibliography}
\end{document}